\documentclass{article}
\usepackage{fullpage}

\usepackage{natbib}

\usepackage{array,xspace,multirow,hhline,tikz,tabularx,booktabs,fixltx2e}
\usepackage{pgflibraryshapes}
\usetikzlibrary{fit} 
\usepackage{graphicx}
\usepackage{epstopdf}
\usepackage{subcaption}
\usepackage{mathtools}
\usepackage{hhline}
\DeclareGraphicsExtensions{.eps}

\usepackage{enumitem}

\usepackage{authblk}
\usepackage{amsmath,amssymb,amsfonts,amsthm}

\makeatletter
\newtheorem*{rep@theorem}{\rep@title}
\newcommand{\newreptheorem}[2]{%
\newenvironment{rep#1}[1]{%
 \def\rep@title{#2 \ref{##1}}%
 \begin{rep@theorem}}%
 {\end{rep@theorem}}}
\makeatother
\newreptheorem{theorem}{Theorem}
\newreptheorem{corollary}{Corollary}

\usepackage{algorithm,algorithmic}

\newcommand{\romanenumi}{\renewcommand{\theenumi}{({\emph{\roman{enumi}}})}
   \renewcommand{\labelenumi}{\theenumi}}

\newcommand{\ie}{i.e.,\xspace}
\newcommand{\eg}{e.g.,\xspace}

\newtheorem{definition}{Definition}%
\newtheorem{theorem}[definition]{Theorem}%
\newtheorem{lemma}[definition]{Lemma}%
\newtheorem{proposition}[definition]{Proposition}%
\newtheorem{observation}{Observation}%
\newtheorem{claim}[definition]{Claim}%
\newcommand{\secref}[1]{Section~\ref{#1}}
\newcommand{\figref}[1]{Figure~\ref{#1}}

\newcommand{\lemref}[1]{Lemma~\ref{#1}}
\newcommand{\thmref}[1]{Theorem~\ref{#1}}

\newcommand{\defref}[1]{Definition~\ref{#1}}
\newcommand{\propref}[1]{Proposition~\ref{#1}}

\newcommand{\obsref}[1]{Observation~\ref{#1}}

\newcommand{\dom}{\overline{D}}
\newcommand{\set}[1]{\{#1\}}
\newcommand{\midd}{\mathrel{:}}

\newcommand{\es}[0]{\ensuremath{\mathit{ES}}\xspace}
\newcommand{\bp}[0]{\ensuremath{\mathit{BP}}\xspace}
\newcommand{\supp}[0]{\ensuremath{\mathit{supp}}\xspace}

\title{Computing Possible and Necessary Equilibrium Actions\\(and Bipartisan Set Winners)\thanks{A short version of this article has appeared in the \textit{Proceedings of the 30th AAAI Conference on Artificial Intelligence} \citep{BFC16a}.}}

\author[1]{Markus Brill}
\author[2]{Rupert Freeman}
\author[3]{Vincent Conitzer}

\affil[1]{Efficient Algorithms Research Group, TU Berlin, Germany}
\affil[2]{Darden School of Business, University of Virginia, USA}
\affil[3]{Department of Computer Science, Duke University, USA}

\date{}

\begin{document}

\maketitle

\begin{abstract}
\noindent In many multiagent environments, a designer has some, but limited control over the game being played. In this paper, we formalize this by considering incompletely specified games, in which some entries of the payoff matrices can be chosen from a specified set. We show that it is NP-hard for the designer to make these choices optimally, even in zero-sum games.  In fact, it is already intractable to decide whether a given action is (potentially or necessarily) played in equilibrium. 
We also consider incompletely specified symmetric games in which all completions are required to be symmetric. Here, hardness holds even in weak tournament games (symmetric zero-sum games whose entries are all $-1$, $0$, or $1$) 
and in tournament games 
(symmetric zero-sum games whose non-diagonal entries are all $-1$ or~$1$). 
The latter result settles the complexity of the possible and necessary winner problems for a social-choice-theoretic solution concept known as the bipartisan set. 
We finally give a mixed-integer linear programming formulation for weak tournament games and evaluate it experimentally.
\end{abstract}

\section{Introduction}

Game theory provides the natural toolkit for reasoning about systems of
multiple self-interested agents.  In some cases, the game is exogenously
determined and all that is left to do is to figure out how to play it.  For
example, if we are trying to solve heads-up limit Texas hold'em poker (as
was effectively done
by~\citet{BBJT15a}), there is no question about
what the game is.  Out in the world, however, the rules of the game are
generally not set in stone.  Often, there is an agent, to whom we will
refer as the {\em designer} or {\em principal}, that has some control over
the game played.  Consider, for example, applications of game theory to
security
domains~\citep{Pita08:Using,Tsai09:IRIS,An12:PROTECT,Yin12:TRUSTSAIMAG}.  In
the long run, the game could be changed, by adding or subtracting security
resources~\citep{Bhattacharya11:Approximation} or reorganizing the targets
being defended (roads, flights, etc.).

{\em Mechanism design} constitutes the extreme case of this, where the
designer typically has complete freedom in choosing the game to be played
by the agents (but still faces a challenging problem due to the agents'
private information).  However, out in the world, we generally also do not
find this other extreme.  Usually, some existing systems are in place and
place constraints on what the designer can do.  This is true to some extent
even in the contexts where mechanism design is most fruitfully applied.
For example, one can imagine that it would be difficult and costly for a
major search engine to entirely redesign its existing auction mechanism for
allocating advertisement space, because of existing users' expectations,
interfacing software, etc.  But this does not mean that aspects of the game
played by the advertisers cannot be tweaked in the designer's favor.

In this paper, we introduce a general framework for addressing intermediate
cases, where the designer has some but not full control over the game.  We
focus on {\em incompletely specified games}, where some entries of the game
matrix contain {\em sets} of payoffs, from among which the designer must
choose.  The designer's aim is to choose so that the resulting equilibrium
of the game is desirable to her.  
This problem is conceptually related to $k$-implementation~\citep{Monderer04:K} and the closely related internal
implementation~\citep{Anderson10:Internal}, where one of the parties is
also able to modify an existing game to achieve better equilibria for
herself. However, in those papers the game is modified by committing to
payments, whereas we focus on choosing from a fixed set of payoffs in an entry.

We focus on two-player zero-sum games, both symmetric and not (necessarily) symmetric, and show NP-hardness in both cases. (Due to a technical reason explained later, hardness for the symmetric case does not imply hardness for the not-necessarily-symmetric case.) The hardness result for the symmetric case holds even for {\em weak tournament games}, in which the payoffs are all $-1$, $0$, or $1$, and for {\em tournament games}, in which the off-diagonal payoffs are all $-1$ or $1$. 

These results have direct implications for related problem in {\em computational social choice}, another important research area in multiagent systems. In social choice (specifically, voting), we take as input a vector of rankings of the alternatives (\eg $a \succ c \succ b$) and as output return some subset of the alternatives. Some social choice functions are based on the {\em pairwise majority graph} which has a directed edge from one alternative to another if a majority of voters prefers the former. 
One attractive concept is that of the {\em essential set}~\citep{LLL93b,DuLa99a}, which can be thought of as based on the following weak tournament game. Two abstract players simultaneously pick an alternative, and if one player's chosen alternative has an edge to the other's, the former wins. Then, the essential set ($\es$) consists of all alternatives that are played with positive probability in some equilibrium. In the absence of majority ties, this game is a tournament game and its essential set is referred to as the {\em bipartisan set} ($\bp$). 

An important computational problem in social choice is the {\em possible (necessary) winner problem}~\citep{Conitzer02:Elicitation,KoLa05a,LPR+11a,XiCo11a,ABF+15a}: given only {\em partial} information about the voters' preferences---for example, because we have yet to elicit the preferences of some of the voters---is a given alternative potentially (necessarily) one of the chosen ones?
It can thus be seen that our hardness results for incompletely specified (weak) tournament games directly imply hardness for the possible/necessary winner problems for $\es$ and $\bp$.

We conclude the paper by formulating and evaluating the efficacy of a
mixed-integer linear programming formulation for the possible equilibrium action problem in weak tournament games. 
Due to the space constraint, most (details of) proofs have been omitted and can be found in the full version of this paper.

\section{Examples}

The following is an incompletely specified two-player symmetric zero-sum game with actions $a,b,c,d$.

\newcommand{\WidestEntry}{$\set{-1,0,1}$}%
\newcommand{\SetToWidest}[1]{\makebox[\widthof{\WidestEntry}]{#1}}%

\vspace{-.3em}
\begin{center}
	\renewcommand{\arraystretch}{1.2}
	  $\begin{array}{r|cccc|}
	\multicolumn{1}{c}{}&\multicolumn{1}{c}{a}&\multicolumn{1}{c}{b}&\multicolumn{1}{c}{c}&\multicolumn{1}{c}{d}\\
	    \cline{2-5}
		a & 0 & \SetToWidest{1} & 0 &  \set{-1,0,1} \\
	    b & -1 & 0 & 1 & 0  \\
	    c &  0 & -1 & 0 & 1 	\\
		d & \set{-1,0,1} & 0 & -1 & 0 \\
		\cline{2-5}
	  \end{array}$
\end{center}
\vspace{.25em}
Here, each entry specifies the payoff to the row player (since the game is
zero-sum, the column player's payoff is implicit), and the set notation
indicates that the payoff in an entry is not yet fully specified. E.g., $\{-1,0,1\}$ indicates that the designer may choose either $-1$, $0$, or
$1$ for this entry.
In the case of symmetric games, we require that the designer keep the game
symmetric, so that if she sets\footnote{Let $u_r(x,y)$ denote the payoff to the row player in row $x$ and column $y$.} $u_r(d,a)=1$ then she must also
set $u_r(a,d)=-1$. Thus, our example game has three possible completions.
The goal for the designer, then, is to choose a completion in such a way
that the equilibrium of the resulting game is desirable to her. For
example, the designer may aim to have only actions~$a$ and~$c$ played with positive
probability in equilibrium. Can she set the payoffs so that this happens?
The answer is yes, because the completion with $u_r(a,d)=1$ has this property.
Indeed, for any $p \ge \frac{1}{2}$, the mixed strategy $pa+(1-p)c$ is an equilibrium strategy for this completion (and no other equilibrium strategies exist).
On the other hand, the completion with $u_r(a,d)=-1$ does have Nash equilibria in which $b$ and $d$ are played with positive probability (for example, both players mixing uniformly is a Nash equilibrium of this game).

Next, consider the following incompletely specified asymmetric zero-sum game:
\vspace{-.5em}
\begin{center}
	\renewcommand{\arraystretch}{1.2}
	\renewcommand{\WidestEntry}{$\set{-1,1}$}%
	  $\begin{array}{r|cc|}
	\multicolumn{1}{c}{}&\multicolumn{1}{c}{\ell}&\multicolumn{1}{c}{r}\\
	    \cline{2-3}
  	  t & -2 & \SetToWidest{1} \\
  	  b & \set{-1,1} & 0\\
		\cline{2-3}
	  \end{array}$
\end{center}
\vspace{.25em}
Suppose the designer's goal is to \emph{avoid} row $t$ being played in
equilibrium.  One might think that the best way to achieve this is to make
row $b$ (the only other row) look as good as possible, and thus set
$u_r(b,\ell)=1$. This results in a fully mixed equilibrium where~$t$ is
played with probability~$\frac{1}{4}$ (and~$\ell$ with~$\frac{1}{4}$).  On
the other hand, setting $u_r(b,\ell)=-1$ results in $\ell$ being a strictly
dominant strategy for the column player, and thus the row player would actually play
$b$ with probability $1$.

\section{Preliminaries}
\label{sec:prelims}

In this section, we formally introduce the concepts 
and computational problems studied in the paper.
For a natural number $n$, let $[n]$ denote the set $\set{1,\ldots,n}$.

\subsection{Games}

A matrix $M \in \mathbb{Q}^{m \times n}$ defines a two-player zero-sum game
(or \emph{matrix game}) as follows. Let the rows of $M$ be indexed by
$I=[m]$ and the columns of $M$ be indexed by $J=[n]$, so that $M =
(m(i,j))_{i \in I, j \in J}$. Player $1$, the \emph{row player}, has action
set $I$ and player $2$, the \emph{column player}, has action set~$J$. If
the row player plays action $i \in I$ and the column player plays action $j
\in J$, the payoff to the row player is given by $m(i,j)$ and the payoff to
the column player is given by $-m(i,j)$. A \emph{(mixed) strategy} of the
row (resp., column) player is a probability distribution over $I$ (resp.,
$J$).
Payoffs are
extended to mixed strategy profiles in the usual way.

A matrix game $M=(m(i,j))_{i \in I, j \in J}$ is \emph{symmetric} if $I=J$
and $m(i,j)= -m(j,i)$ for all $(i,j) \in I \times J$. 
A \emph{weak tournament game} is a symmetric matrix game in which all
payoffs are from the set $\set{-1,0,1}$.  Weak tournament games naturally
correspond to directed graphs $W=(A,\succ)$ as follows: vertices
correspond to actions and there is a directed edge from
action $a$ to action $b$ (denoted $a \succ b$) if and only if the payoff to the row player in
action profile $(a,b)$ is~$1$. 
A \emph{tournament game} is a weak tournament game with the additional property that the payoff is $0$ \emph{only if} both players choose the same action. The corresponding graph thus has a directed edge for every pair of (distinct) vertices.

\subsection{Incomplete Games}

An \emph{incompletely specified matrix game} (short: \emph{incomplete
  matrix game}) is given by a matrix $M \in (2^\mathbb{Q})^{m \times
  n}$. That is, every entry of the matrix $M = (m(i,j))_{i \in I, j \in J}$
is a subset $m(i,j) \subseteq \mathbb{Q}$. If $m(i,j)$ consists of a single
element, 
we say that the payoff for action profile $(i,j)$ is \emph{specified}, and
write $m(i,j)=m$ instead of the more cumbersome $m(i,j)=\set{m}$. 
For an incomplete matrix game, the set of \emph{completions} is given by the 
set of all matrix games that arise from selecting a number from the corresponding set for every action profile for which the payoff is unspecified.

An \emph{incomplete symmetric game} is an incomplete matrix game with $m(j,i) =
\set{-m: m \in m(i,j)}$ for all $i \in I$ and $j \in J$. The set of
\emph{symmetric completions} of an incomplete symmetric game is given
by the set of all completions that are symmetric.  When considering
incomplete symmetric games, we will restrict attention to symmetric
completions, which is the reason hardness results do not transfer from the
symmetric case to the general case.
An \emph{incomplete weak tournament game} is an incomplete symmetric game for which 
(1) every unspecified payoff has the form $m(i,j)=\set{-1,0,1}$ with $i \neq j$, and 
(2) every symmetric completion is a weak tournament game. 
An \emph{incomplete tournament game} is an incomplete symmetric game for which 
(1) every unspecified payoff has the form $m(i,j)=\set{-1,1}$ with $i \neq j$, and 
(2) every symmetric completion is a tournament game. 
Every incomplete (weak) tournament game corresponds to a directed graph in which the relation for certain pairs $(i,j)$ of distinct vertices is unspecified. Whereas every completion of an incomplete tournament game satisfies either $m(i,j)=1$ or $m(i,j)=-1$ for any such pair, a completion of an incomplete \emph{weak} tournament game also allows ``ties,'' \ie $m(i,j)=0$.

\subsection{Equilibrium Concepts}

The standard solution concept for normal-form games is Nash equilibrium. A strategy profile $(\sigma,\tau)$ is a Nash equilibrium of a matrix game $M$ if the strategies $\sigma$ and $\tau$ are best responses to each other, \ie $m(\sigma,j) \ge m(\sigma,\tau) \ge m(i,\tau)$ for all $i \in I$ and $j \in J$. The payoff to the row player is identical in all Nash equilibria, and is known as the \emph{value} of the game. 

We are interested in the question whether an action is played with positive
probability in at least one Nash equilibrium. 
For improved readability, the following definitions are only formulated for the row player; definitions for the column player are analogous.
The support $\supp(\sigma)$ of a strategy $\sigma$ is the set of actions that are played with positive probability in $\sigma$.
An action $i \in I$ is called \emph{essential} 
if there exists a Nash equilibrium $(\sigma,t)$ with $i \in \supp(\sigma)$. By $\es_\text{row}(M)$ we denote the set of all actions $i \in I$ that are essential.

\begin{definition}
	The \emph{essential set} $\es(M)$ of a matrix game~$M$ contains all actions that are essential, \ie $\es(M) = \es_\text{row}(M) \cup \es_\text{column}(M)$.
\end{definition}  

There is a useful connection between the essential set and \emph{quasi-strict} (Nash) equilibria. Quasi-strictness is a refinement of Nash equilibrium that requires that every best response is played with positive probability \citep{Hars73a}. Formally, a Nash equilibrium $(\sigma,t)$ of a matrix game $M$ is \emph{quasi-strict} if $m(\sigma,j) > m(\sigma,\tau) > m(i,\tau)$ for all $i \in I \setminus \supp(\sigma)$ and $j \in J \setminus \supp(\tau)$. 
Since the set of Nash equilibria of a matrix game $M$ is convex, there
always exists a Nash equilibrium $(\sigma,\tau)$ with $\supp(\sigma) \cup
\supp(\tau) = \es(M)$. Moreover, it has been shown that all quasi-strict
equilibria of a matrix game have the same support \citep{BrFi08a}.  
Thus, an action is contained in the essential set of a matrix game if and only if it is played with positive probability in some quasi-strict Nash equilibrium.
\citet{BrFi08a} have shown that quasi-strict equilibria, and thus the essential set, can be computed in polynomial time.

\subsection{Computational Problems}
\label{sec:comp}

\newcommand{\problem}{Equilibrium Containment Problem\xspace}
\newcommand{\prob}{ECP\xspace}

We are interested in the computational complexity of the following decision problems.
\begin{itemize}
	\item \textbf{Possible Equilibrium Action:}
Given an incomplete matrix game $M$ and an action $a$, is there a completion $M'$ of $M$ such that $a \in \es(M')$?
	\item \textbf{Necessary Equilibrium Action:}
Given an incomplete matrix game $M$ and an action $a$, is it the case that $a \in \es(M')$ for all completions $M'$ of $M$?
\end{itemize}
One may wonder why these are the right problems to solve.
Most generally, the designer could have a utility for each possible outcome (\ie action profile) of the game.  The next proposition shows that hardness of the possible
equilibrium action problem immediately implies hardness of the problem of maximizing
the designer's utility.
\begin{proposition}
  Suppose
the possible equilibrium action problem is NP-hard.
Then, if the designer's payoffs are nonnegative, no
  positive approximation guarantee for the designer's utility (in the
  optimistic model where the best equilibrium for the designer is chosen) can be given
  in polynomial time unless P=NP.
\end{proposition}
\begin{proof}
  Suppose, for the sake of contradiction, that there is a polynomial time
  algorithm with a positive approximation guarantee for the problem of
  maximizing the designer's optimistic utility.  Then we can use this
  algorithm for determining whether there is a completion where a strategy
  receives positive probability in some equilibrium: simply give the
  designer utility $1$ for all outcomes in which that strategy is played,
  and $0$ everywhere else.  The designer can get the strategy to be played
  with positive probability if and only if she can obtain positive utility
  from this game, and she can obtain positive utility from this game if and
  only if the approximation algorithm returns a positive utility.
\end{proof}
A similar connection can be given between the necessary equilibrium action
problem and the case where designer utilities are
nonpositive and a pessimistic model is used (assigning a payoff of $-1$
to the action in question and $0$ otherwise).

In the context of weak tournament games, the essential set ($\es$) is often interpreted as a (social) choice function identifying desirable alternatives \citep{DuLa99a}. In the special case of tournament games, the essential set is referred to as the \emph{bipartisan set} ($\bp$) \citep{LLL93b}. 
The possible and necessary equilibrium action problems defined above thus correspond to possible and necessary \emph{winner} queries for the social choice functions $\es$ (for weak tournament games) and $\bp$ (for tournament games).
The computational complexity of possible and necessary winners has been studied for many common social choice functions (\eg \citealp{XiCo11a}; \citealp{ABF+15a}). To the best of our knowledge, we are the first to provide complexity results for $\es$ and $\bp$.

\section{Zero-Sum Games}
\label{sec:zerosum}

\newcommand{\cyc}{alternating\xspace}
\newcommand{\tr}{\mathit{tr}}

In this section, we show that computing possible and necessary equilibrium actions is intractable for (not-necessarily-symmetric) matrix games. In the proofs, we will make use of a class of games that we call \emph{\cyc games}. Intuitively, an \cyc
game is a generalized version of Rock-Paper-Scissors that additionally
allows ``tiebreaking payoffs'' which are small payoffs in cases where both
players play the same action. A formal definition and proofs of some required properties are given in the full version of this paper.

We first consider the \emph{necessary} equilibrium action problem.
For ease of readability, we only give an informal proof sketch here. Much of the work in the complete proof (to be found in the appendix) is to correctly set 
values for constants so that the desired equilibrium properties hold.

\begin{theorem} \label{thm:matrix-nec} 
	The necessary equilibrium action problem (in matrix games that are not
        necessarily symmetric) is {coNP-complete}.
\end{theorem}

\begin{proof}[Proof sketch]

	For NP-hardness, we give a reduction from \textsc{SetCover}. An instance of \textsc{SetCover} is given by a collection $\{ S_1, \ldots, S_n \}$ of subsets of a universe $U$, and an integer~$k$; the question is whether we can cover $U$ using only $k$ of the subsets. 
We may assume that $k$ is odd (it is always possible to add a singleton
subset with an element not covered by anything else and increase $k$ by $1$). Define an incomplete matrix game $M$ where the row player has $3n-k+1$ actions, and the column player has $2n-k+|U|$ actions. The row player's actions are given by $\{ S_{i,j} : i \in [n], j \in [2] \} \cup \{ x_i : i \in [n-k]\} \cup \{ r_* \}$.

		\begin{figure} 
		\renewcommand{\arraystretch}{1.8}
		\begin{center}
\footnotesize
\setlength{\tabcolsep}{0.19em}
\scalebox{1.2}{
\begin{tabular}{r|c|c|c|c|c||c|c|c|c|c|}
		
		    		  \multicolumn{1}{c}{} & \multicolumn{1}{c}{$c_1$} & \multicolumn{1}{c}{$c_2$} & \multicolumn{1}{c}{$c_3$} & \multicolumn{1}{c}{$c_4$} & \multicolumn{1}{c}{$c_5$} & \multicolumn{1}{c}{$s_1$} & \multicolumn{1}{c}{$s_2$} & \multicolumn{1}{c}{$s_3$} & \multicolumn{1}{c}{$s_4$} & \multicolumn{1}{c}{$s_5$}  \\ \cline{2-11}
		    	    $S_{1,1}$ & $0$ & $\phantom{\scalebox{0.75}[1.0]{\( - \)}}H$ & $\scalebox{0.75}[1.0]{\( - \)}H$ & $\phantom{\scalebox{0.75}[1.0]{\( - \)}}H$ & \phantom{h} $\scalebox{0.75}[1.0]{\( - \)}H$ \phantom{h} & $y$ & $y$ & $0$ & $0$ & $0$  \\ \cline{2-11}
		    	    $S_{1,2}$ & $\{\scalebox{0.75}[1.0]{\( - \)}1, 1 \}$ & $\phantom{\scalebox{0.75}[1.0]{\( - \)}}H$ & $\scalebox{0.75}[1.0]{\( - \)}H$ & $\phantom{\scalebox{0.75}[1.0]{\( - \)}}H$ & $\scalebox{0.75}[1.0]{\( - \)}H$ & $x$ & $x$ & $0$ & $0$ & $0$  \\  \cline{2-11}
		    	    $S_{2,1}$ & $\scalebox{0.75}[1.0]{\( - \)}H$ & $0$ & $\phantom{\scalebox{0.75}[1.0]{\( - \)}}H$ & $\scalebox{0.75}[1.0]{\( - \)}H$ & $\phantom{\scalebox{0.75}[1.0]{\( - \)}}H$ & $0$ & $0$ & $y$ & $y$ & $0$ \\          \cline{2-11}
		    	    $S_{2,2}$ & $\scalebox{0.75}[1.0]{\( - \)}H$ & $\{\scalebox{0.75}[1.0]{\( - \)}1, 1 \}$ & $\phantom{\scalebox{0.75}[1.0]{\( - \)}}H$ & $\scalebox{0.75}[1.0]{\( - \)}H$ & $\phantom{\scalebox{0.75}[1.0]{\( - \)}}H$ & $0$ & $0$ & $x$ & $x$ & $0$\\ \cline{2-11}
		    	    $S_{3,1}$ & $\phantom{\scalebox{0.75}[1.0]{\( - \)}}H$ & $\scalebox{0.75}[1.0]{\( - \)}H$ & $0$ & $\phantom{\scalebox{0.75}[1.0]{\( - \)}}H$ & $\scalebox{0.75}[1.0]{\( - \)}H$ & $0$ & $0$ & $y$ & $0$ & $y$\\          \cline{2-11}
		    	    $S_{3,2}$ & $\phantom{\scalebox{0.75}[1.0]{\( - \)}}H$ & $\scalebox{0.75}[1.0]{\( - \)}H$ & $\{\scalebox{0.75}[1.0]{\( - \)}1, 1 \}$ & $\phantom{\scalebox{0.75}[1.0]{\( - \)}}H$ & $\scalebox{0.75}[1.0]{\( - \)}H$ & $0$ & $0$ & $x$ & $0$ & $x$ \\ \cline{2-11}
		    	    $S_{4,1}$ & ${\scalebox{0.75}[1.0]{\( - \)}}H$ & $\phantom{\scalebox{0.75}[1.0]{\( - \)}}H$ & $\scalebox{0.75}[1.0]{\( - \)}H$ & $0$ & $\phantom{\scalebox{0.75}[1.0]{\( - \)}}H$  & $0$ & $y$ & $0$ & $y$ & 0 \\          \cline{2-11}
		    	    $S_{4,2}$ & ${\scalebox{0.75}[1.0]{\( - \)}}H$ & $\phantom{\scalebox{0.75}[1.0]{\( - \)}}H$ & $\scalebox{0.75}[1.0]{\( - \)}H$ & $\{\scalebox{0.75}[1.0]{\( - \)}1, 1 \}$ & $\phantom{\scalebox{0.75}[1.0]{\( - \)}}H$ &  $0$ & $x$ & $0$ & $x$ & $0$\\ \cline{2-11}
		    		$x_1$     & $\phantom{\scalebox{0.75}[1.0]{\( - \)}}H$ & $\scalebox{0.75}[1.0]{\( - \)}H$ & $\phantom{\scalebox{0.75}[1.0]{\( - \)}}H$ & $\phantom{\scalebox{0.75}[1.0]{\( - \)}}H$ & $\scalebox{0.75}[1.0]{\( - \)}1$ & $0$ & $0$ & $0$ & $0$ & $0$  \\ \hhline{~|=|=|=|=|=#-|-|-|-|-|}
\multicolumn{1}{c|}{$r_*$} & \multicolumn{1}{c|}{$-v$} & \multicolumn{1}{c|}{$-v$} & \multicolumn{1}{c|}{$-v$} & \multicolumn{1}{c|}{$-v$} & \multicolumn{1}{c|}{$-v$} & \multicolumn{1}{c|}{G} & \multicolumn{1}{c|}{G} & \multicolumn{1}{c|}{G} & \multicolumn{1}{c|}{G} & \multicolumn{1}{c|}{G} 

\\       \cline{2-11}
		  \end{tabular}
		  }
		\caption{The incomplete matrix game $M$ used in the proof of \thmref{thm:matrix-nec} for the \textsc{SetCover} instance given by {$|U|=5$}, $n=4$, $k=3$, $S_1=\{s_1, s_2 \}$, $S_2 = \{s_3, s_4 \}$, $S_3=\{s_3, s_5 \}$, and $S_4 = \{ s_2, s_4 \}$. $L$ lies in the top left, indicated by double lines.
		}
		\end{center}
		\end{figure}

	            Let $L$ denote the restriction of the game to the first
                $2n-k$ columns and $3n-k$ rows. We denote the column player's actions in
                this part of the game by $c_1, \ldots, c_{2n-k}$. We
                set $m(S_{i,1},c_i)=0$ and $m(S_{i,2},c_i)=\{ -1, 1 \}$ for all $i \in [n]$ and
                $m(x_{i}, c_{n+i})=-1$ for all $i \in [n-k]$. We fill in the remaining entries with
                $H$ and $-H$, where $H$ is a large positive number, so that if we consider only one of each pair of rows $\{ S_{i,1}$, $S_{i,2} \}$, $L$ acts as an \cyc game. Setting $m(S_{i,2},c_i)=-1$ will correspond to choosing $S_i$ for the set cover, and setting $m(S_{i,2},c_i)=1$ will correspond to not choosing $S_i$. Note that, considering only $L$, the row player will put positive probability on exactly one of $S_{i,1}$ and $S_{i,2}$ (as well as all rows $x_i$) and, as long as $H$ is sufficiently large, each row that is played with positive probability receives approximately $\frac{1}{N}$ probability. $S_{i,1}$ is played if $S_i$ is chosen for the set cover, $S_{i,2}$ is played otherwise. Also note that the value of $L$ is close to zero, depending on the exact setting of the undetermined entries.
	
We have additional columns $s_1, \ldots, s_{|U|}$ corresponding to elements of $U$. For every set $S_j$ containing $s_i$, column $s_i$ has a positive entry $y$ in row $S_{j,1}$ and a negative entry $x$ in row $S_{j,2}$. 
If $s_i$ is not covered by any chosen set, then the equilibrium can not be contained in $L$: If it were, then the column player could best respond by playing $s_i$, where all entries (on rows played with positive probability by the row player) are either $0$ or $x<0$. However, if $s_i$ is covered by some set, then we can make $y$ large enough (relative to~$x$) that the column player will not play $s_i$. Thus, if every $s_i$ is covered, the column player plays only columns from $L$.

Finally, we have a single extra row labeled $r_*$. This row has a small negative payoff $-v$ for all columns in $L$, and a very large positive payoff $G$ for all columns not in $L$. As long as the equilibrium is contained in $L$ (\ie all elements are covered), it is not a best response for the row player to play~$r_*$. However, if the column player puts positive probability on some $s_i$ (that is, $s_i$ is uncovered), then $G$ is large enough that the row player can best respond by playing $r_*$ with positive probability.
\end{proof}

By modifying the construction in the proof of \thmref{thm:matrix-nec}, we also get a hardness result for the problem of deciding whether an action is a \emph{possible} equilibrium action.

\begin{theorem} \label{thm:matrix-pos}
	The possible equilibrium action problem (in matrix games that are not
        necessarily symmetric) is NP-complete.
\end{theorem}

\section{Weak Tournament Games}
\label{sec:weak}

We now turn to weak tournament games and analyze the computational complexity of possible and necessary $\es$ winners.

\begin{theorem} \label{thm:weak-pos}
	The possible $\es$ winner problem (in weak tournament games) is NP-complete. 
\end{theorem}

\begin{proof}[Proof sketch]

	For NP-hardness, we provide a reduction from \textsc{Sat}. 
	Let $\varphi = C_1 \land \ldots \land C_m$ be a Boolean formula in conjunctive normal form over a finite set $V= \{v_1, \ldots, v_n\}$ of variables. 
	We define an incomplete weak tournament\footnote{We utilize the one-to-one correspondence between weak tournament games and directed graphs without cycles of length one or two (so-called \emph{weak tournaments}). For a weak tournament $(A,\succ)$, we use the notation $a \succ b$ to denote a directed edge from~$a$~to~$b$.} $W_\varphi=(A,\succ)$ as follows.
	The set $A$ of vertices is given by $A =  \cup_{i=1}^n X_i \cup \set{c_1, \ldots, c_m} \cup \set{d}$, where $X_i = \set{x_i^1,\ldots,x_i^6}$ for all $i \in [n]$. Vertex $c_j$ corresponds to clause $C_j$ and the set $X_i$ corresponds to variable~$v_i$.
	
	Within each set $X_i$, there is a cycle 
$x_i^1 \succ x_i^2 \succ x_i^3 \succ x_i^4 \succ x_i^5 \succ x_i^6 \succ x_i^1$ and an unspecified edge between $x_i^1$ and $x_i^4$. 
	If variable $v_i$ occurs as a positive literal in clause $C_j$, we
        have edges $c_j \succ x_i^3$ and $x_i^5 \succ c_j$. If variable
        $v_i$ occurs as a negative literal in clause $C_j$, we have edges
        $c_j \succ x_i^6$ and $x_i^2 \succ c_j$. Moreover, there is an edge from $c_j$ to $d$ for every $j \in [m]$. For all pairs of vertices for which neither an edge has been defined, nor an unspecified edge declared, we have a tie.  See \figref{fig:construction} for an example.

	We make two observations about $W_\varphi$. 
	
	\begin{observation} \label{obs1}
		For each completion $W$ of $W_\varphi$, we have $d \in \es(W)$ if and only if $\es(W) \cap \set{c_1, \ldots, c_m} = \emptyset$.
	\end{observation}
	
	\begin{observation} \label{obs2}
		For each $i$, there is exactly one unspecified edge within (and thus exactly three possible completions of) the subtournament $W_\varphi|_{X_i}$. If the we set a tie between $x_i^1$ and $x_i^4$, then all Nash equilibria $p$ of the subtournament $W_\varphi|_{X_i}$ satisfy $p(x_i^1)=p(x_i^3)=p(x_i^5)$ and $p(x_i^2)=p(x_i^4)=p(x_i^6)$. If we set $x_i^1 \succ x_i^4$, then every quasi-strict equilibrium $p$ of $W_\varphi|_{X_i}$ satisfies $p(x_i^2)=p(x_i^4)=p(x_i^6)=0$, $p(x_i^5)>p(x_i^1)>p(x_i^3)>0$, and $p(x_i^1)+p(x_i^3)>p(x_i^5)$. 
	By symmetry, setting $x_i^4 \succ x_i^1$ results in quasi-strict equilibria $p$ with $p(x_i^1)=p(x_i^3)=p(x_i^5)=0$, $p(x_i^4)>p(x_i^6)>p(x_i^2)>0$, and $p(x_i^2)+p(x_i^6)>p(x_i^4)$. 
	\end{observation}
	
	We can now show hat $\varphi$ is satisfiable if and only if there is a completion $W$ of $W_\varphi$ with $d \in \es(W)$. For the direction from left to right, let $\alpha$ be a satisfying assignment and consider the completion $W$ of $W_\varphi$ as follows: if $v_i$ is set to ``true'' under $\alpha$, add edge $x_i^1 \succ x_i^4$; otherwise, add edge $x_i^4 \succ x_i^1$.
	It can be shown that $\es(W) = \cup_{i \in [n]} \es(W|_{X_i}) \cup \set{d}$.
	
	For the direction from right to left, let $W$ be a completion of $W_\varphi$ with $d \in \es(W)$. Define the assignment $\alpha$ by setting variable $v_i$ to ``true'' if $x_i^1 \succ x_i^4$ and to ``false'' if $x_i^4 \succ x_i^1$. If there is a tie between $x_i^1$ and $x_i^4$, we set the truth value of $v_i$ arbitrarily.
	Since $d \in ES(W)$, we know by \obsref{obs1} that $c_j \notin \es(W)$ for all $j \in [m]$. It can now be shown that every $c_i$ has an incoming edge from a vertex in $\es(W)$, and that this vertex corresponds to a literal that appears in $C_i$ and that is set to ``true'' under $\alpha$.
\end{proof}

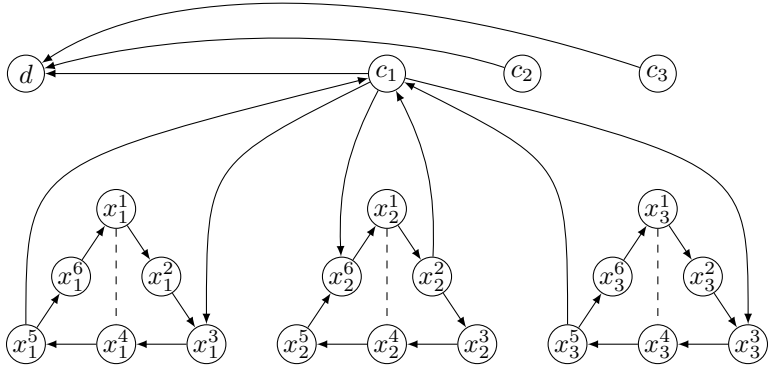
\begin{figure}[tb]
\centering
\begin{tikzpicture}[scale=1.2]
	\tikzstyle{every node}=[circle,draw,minimum size=1.4em,inner sep=0pt]

	\draw(-4,3.5) node(d){$d$};

	\draw (0,3.5) node(c1){$c_1$} ++(1.5,0) node(c2){$c_2$} ++(1.5,0) node(c3){$c_3$};

		\draw[-latex] (c1) -- (d);
		\draw[-latex] (c2) .. controls  (0,4) and (-2.25,4) .. (d);
		\draw[-latex] (c3) .. controls  (0,4.5) and (-2.5,4.5) .. (d);

	\draw (-3,2) node(x11){$x_1^1$}
	++(0.5,-.75) node(x12){$x_1^2$}
	++(0.5,-.75) node(x13){$x_1^3$}
	++(-1,0)     node(x14){$x_1^4$}
	++(-1,0)     node(x15){$x_1^5$}
	++(0.5,.75)  node(x16){$x_1^6$};
	\foreach \x / \y in {x11/x12,x12/x13,x13/x14,x14/x15,x15/x16,x16/x11}
		{ \draw[-latex,] (\x) -- (\y); }
	\draw[dashed] (x11) -- (x14);

	\draw (0,2) node(x21){$x_2^1$}
	++(0.5,-.75) node(x22){$x_2^2$}
	++(0.5,-.75) node(x23){$x_2^3$}
	++(-1,0)     node(x24){$x_2^4$}
	++(-1,0)     node(x25){$x_2^5$}
	++(0.5,.75)  node(x26){$x_2^6$};
	\foreach \x / \y in {x21/x22,x22/x23,x23/x24,x24/x25,x25/x26,x26/x21}
		{ \draw[-latex,] (\x) -- (\y); }
	\draw[dashed] (x21) -- (x24);

	\draw (3,2) node(x31){$x_3^1$}
	++(0.5,-.75) node(x32){$x_3^2$}
	++(0.5,-.75) node(x33){$x_3^3$}
	++(-1,0)     node(x34){$x_3^4$}
	++(-1,0)     node(x35){$x_3^5$}
	++(0.5,.75)  node(x36){$x_3^6$};
	\foreach \x / \y in {x31/x32,x32/x33,x33/x34,x34/x35,x35/x36,x36/x31}
		{ \draw[-latex,] (\x) -- (\y); }
	\draw[dashed] (x31) -- (x34);

	\draw[-latex] (c1) .. controls (-2,2.5)  .. (x13);
	\draw[-latex] (x15) .. controls (-4,2.5)  .. (c1);

	\draw[-latex] (c1) to [bend right=15] (x26);
	\draw[-latex] (x22) to [bend right=15] (c1);

	\draw[-latex] (c1) .. controls (4,2.5)  .. (x33);
	\draw[-latex] (x35).. controls (2,2.5)  .. (c1);

\end{tikzpicture}
\caption{The weak tournament $W_\varphi$ for formula $\varphi= C_1 \wedge C_2 \wedge C_3$ with $C_1 = x_1 \lor \lnot x_2 \lor x_3$. Dashed lines indicate unspecified edges. For improved readability, edges connecting $c_2$ and $c_3$ to $X$ have been omitted.}
\label{fig:construction}
\end{figure}

We get hardness for the necessary winner problem by slightly modifying the construction used in the proof above.

\begin{theorem} \label{thm:weak-nec}
	The necessary $\es$ winner problem (in weak tournament games) is coNP-complete. 
\end{theorem}

It can actually be shown that the problems considered in Theorems \ref{thm:weak-pos} and \ref{thm:weak-nec} remain intractable even in the case where unspecified payoffs can be chosen from the interval $[-1,1]$ (while still maintaining symmetry). This is interesting insofar as this is our only hardness result for \emph{infinite and/or continuous} payoff sets; such relaxations often make problems computationally easier.

\begin{proposition} \label{prop:continuous}
	The possible equilibrium action problem in weak tournament games remains NP-complete (and the necessary equilibrium action problem coNP-complete) when every payoff set $m(i,j)$ with $m(i,j)=\set{-1,0,1}$ is replaced by $m(i,j)=[-1,1]$.
\end{proposition}

\section{Tournament Games}
\label{sec:tournaments}

The hardness results in \secref{sec:weak} leave open the possibility that computing $\es$ is tractable in tournament games (where $\es$ is referred to as $\bp$). Indeed, it often turns out that computational problems become easier to solve when restricting attention to tournaments \citep{BBFHR11a,BrFi08b,KMSc07a}.\footnote{There are also cases in the literature where a computational problem remains hard when restricted to tournaments, but the hardness proof is much more complicated \citep{Alon06a,CTY07a,Coni06a}.}
The reason is that certain structural properties only hold in tournaments.\footnote{For example, \citet{LLL93a} and \citet{FiRy92a} have shown that every tournament game $T$ has a \emph{unique} Nash equilibrium. This Nash equilibrium is quasi-strict and has support $\es(T)=\bp(T)$.}

Nevertheless, we prove that computing possible and necessary $\es$ winners is hard even in tournament games. The technical difficulty in proving these results lies in the fact that the hardness reduction cannot use ``ties'' (\ie non-edges) in (the specified part of) the graph.

\begin{theorem} \label{thm:posBP}
	The possible $\bp$ winner problem (in tournament games) is NP-complete. 
\end{theorem}

\begin{theorem} \label{thm:necBP}
	The necessary $\bp$ winner problem (in tournament games) is coNP-complete. 
\end{theorem}

Observe that these results neither imply the results for weak tournament games in \secref{sec:weak} (where completions can use ``ties'') nor the results for general matrix games in \secref{sec:zerosum} (where completions can be asymmetric).

\section{MIP for Weak Tournament Games}
\label{sec:mip}

Of course, the fact that a problem is NP-hard does not make it go away; it
is still desirable to find algorithms that scale reasonably well (or very
well on natural instances).  NP-hard problems in game theory often allow
such algorithms.  In particular, formulating the problem as a \emph{mixed-integer
program (MIP)} and calling a general-purpose solver often provides good results.
In this section, we formulate the possible $\es$ winner problem in weak tournament games as a MIP.

\subsection{Mixed-Integer Programming Formulation}

Let $W=(w(i,j))_{i,j \in A}$ be an incomplete weak tournament game. 
For every entry $w(i,j)$ of $W$, we define two binary variables $x_{ij}^{\text{pos}}$ and $x_{ij}^{\text{neg}}$.
Setting $w(i,j)$ to $w_{ij} \in \set{-1,0,1}$ corresponds to setting $x_{ij}^{\text{pos}}$ and $x_{ij}^{\text{neg}}$ in such a way that $(x_{ij}^{\text{pos}}, x_{ij}^{\text{neg}}) \neq (1,1)$ and $x_{ij}^{\text{pos}} - x_{ij}^{\text{neg}} = w_{ij}$.
For each action $j$, there is a variable $p_j$ corresponding to the probability that the column player assigns to $j$. 
Finally, $z_{ij}$ is a variable that, in every feasible solution, equals $w_{ij} p_j$.

To determine whether an action $k \in A$ is a possible $\es$ winner of $W$, we solve the following MIP. Every feasible solution of this MIP corresponds to a completion of $W$ and a Nash equilibrium of this completion.
  \begin{alignat*}{3}
	  &\text{maximize } p_k && \\
	  &\text{subject to} && \\
	   	&x_{ij}^{\text{neg}}-x_{ji}^{\text{pos}}=0, \, \forall i,j 	& \quad 		   		&x_{ij}^{\text{pos}}=1, \, \text{if } w(i,j)=1\\
		&x_{ij}^{\text{pos}} + x_{ij}^{\text{neg}} \le 1, \,  \forall i,j & \quad 			&x_{ij}^{\text{neg}}=1, \, \text{if } w(i,j)=-1 \\
		&x_{ij}^{\text{pos}}=x_{ij}^{\text{neg}}=0, \, \text{if } w(i,j)=0 & \quad 			&x_{ij}^{\text{pos}}, x_{ij}^{\text{neg}} \in \set{0,1}, \, \forall i,j\\
		&z_{ij} \ge p_j-2(1-x_{ij}^{\text{pos}}), \, \forall i,j & 			   				&\textstyle{\sum_{j \in A} z_{ij}} \le 0,  \, \forall i \\
		&z_{ij} \ge -p_j-2(1-x_{ij}^{\text{neg}}), \, \forall i,j & 			   			&\textstyle{\sum_{j \in A}p_j} = 1\\
		&z_{ij} \ge -2x_{ij}^{\text{pos}} - 2x_{ij}^{\text{neg}}, \, \forall i,j &			&p_j \ge 0, \, \forall j\\
   \end{alignat*}

\vspace{-0.5cm}
  Here, indices $i$ and $j$ range over the set $A$ of actions.
  Most interesting are the constraints on $z_{ij}$; we note that exactly
  one of the three will be binding depending on the values of
  $x_{ij}^{\text{pos}}$ and $x_{ij}^{\text{neg}}$. The net effect of these constraints is
  to ensure that $z_{ij} \ge w_{ij}p_j$. (Since we also have the constraint $\sum_{j \in A} z_{ij} \le 0$ and because the value of every completion is zero, $z_{ij} = w_{ij}p_j$ in every feasible solution.)
  All other constraints containing $x_{ij}^{\text{pos}}$ or $x_{ij}^{\text{neg}}$ are to impose symmetry and consistency on the entries. The remaining constraints make sure that $p$ is a well-defined probability distribution and that no row yields positive payoff for player $1$. 
  
  It is possible to adapt this MIP to compute possible and necessary $\bp$ winners in tournament games. All that is required is to replace inequality constraints of the form $x_{ij}^{\text{pos}} + x_{ij}^{\text{neg}} \le 1$ by equalities, thus eliminating the possibility to set $w_{ij}=0$. Since tournament games have a unique equilibrium, checking whether action $k$ is a possible or necessary $\bp$ winner can be done by maximizing and minimizing the objective function $p_k$, respectively. The reason that this approach does not extend to the computation of necessary winners in \emph{weak} tournament games is that weak tournaments may have multiple equilibria, some of them not quasi-strict. Since our MIP optimizes over the set of all (not necessarily quasi-strict) equilibria, we may encounter cases where the MIP finds a completion with $p_k=0$, but $k$ is still a necessary winner because it  is played with positive probability in every \emph{quasi-strict} equilibrium.

\subsection{Experimental Results}

\begin{figure}[tb]
	\centering
    \includegraphics[width=7cm]{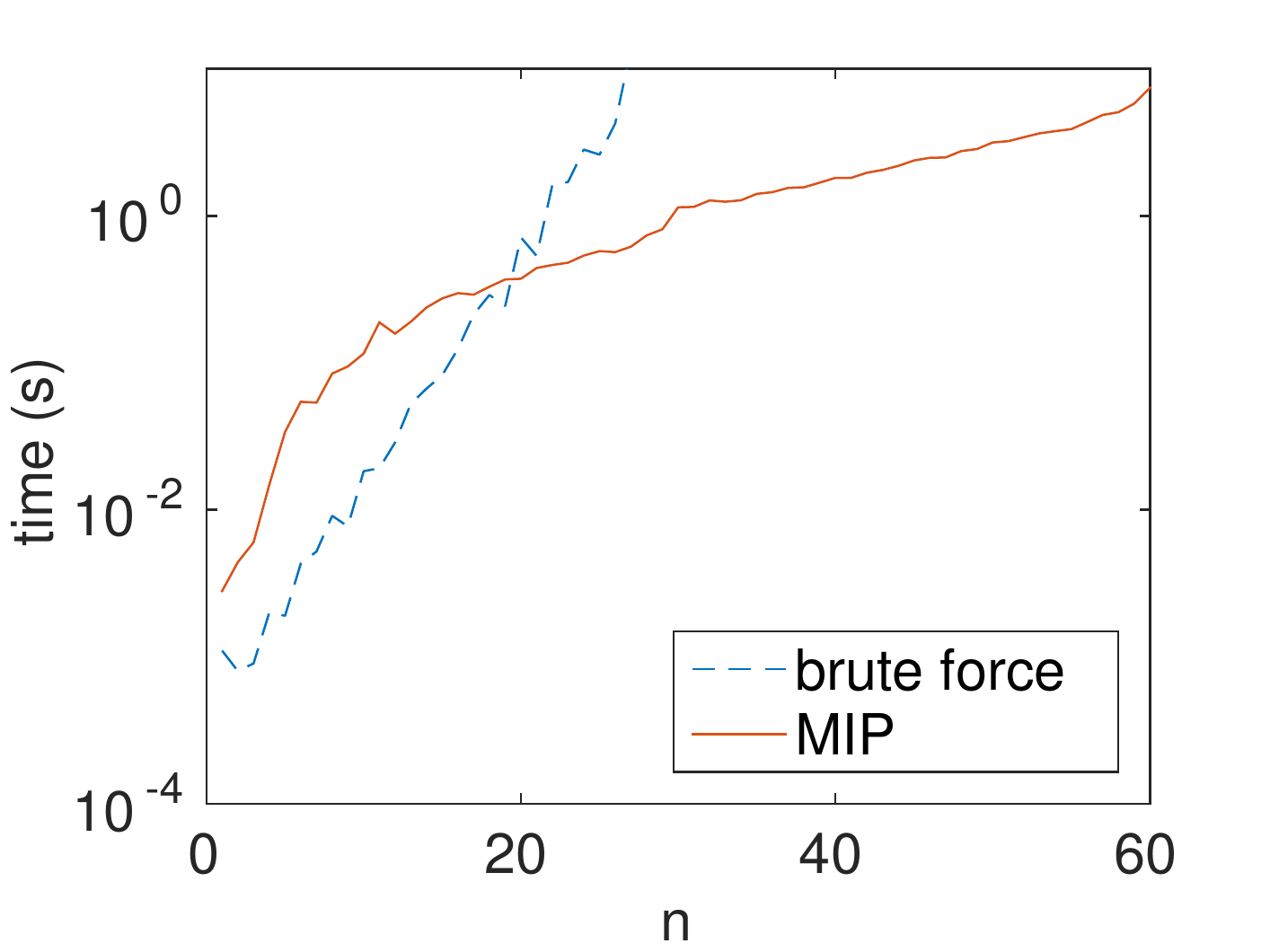}
	\hfill
	\includegraphics[width=7cm]{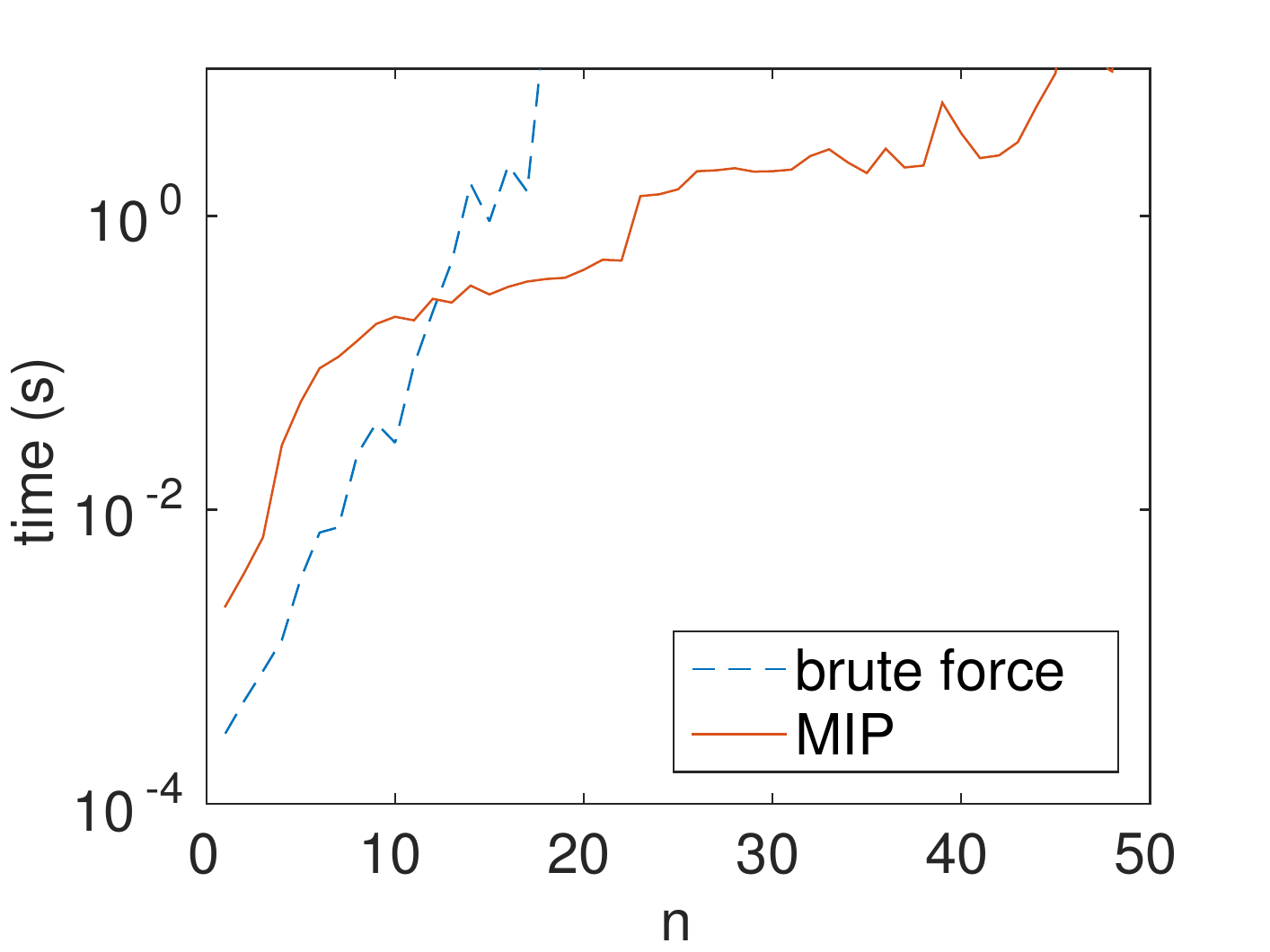}
  \caption{Average runtime (log scale) for $\frac{n}{2}$ unspecified entries (left) and $n$ unspecified entries (right).}
  \label{fig:graphs}
\end{figure}

We tested our MIP for the possible $\es$ winner problem in weak tournament games containing either $\frac{n}{2}$ or $n$ unspecified entries, where $n=|A|$ is the number of actions available to each player. For each $n$, we examined the average time required to solve 100 random instances\footnote{Random instances were generated by randomly choosing each entry from $\{-1,0,1\}$ and imposing symmetry, then randomly choosing the fixed number of entries to be unspecified.} of size~$n$, using CPLEX 12.6 to solve the MIP.  Results are shown in \figref{fig:graphs}, with algorithms cut off once the average time to find a solution exceeds 10 seconds.

We compared the performance of our MIP with a simple brute
force algorithm. The brute force algorithm performs a  depth-first search over the space of all completions, terminating when it finds a certificate of a yes instance or after it has exhausted all completions.
We observe that for even relatively small values of $n$, the MIP begins to
significantly outperform the brute-force algorithm.

\section{Conclusion}
\label{sec:conclusion}

Often, a designer has some, but limited, control over the game being
played, and wants to exert this control to her advantage.  In this paper,
we studied how computationally hard it is for the designer to decide
whether she can choose payoffs in an incompletely specified game to achieve
some goal in equilibrium, and found that this is NP-hard even in quite
restricted cases of two-player zero-sum games.  
Our framework and our results also apply in cases where there is no designer but we are just uncertain about the payoffs, either because further exploration is needed to determine what they are, or because they vary based on conditions (e.g., weather). In such settings one might simply be interested in potential and unavoidable equilibrium outcomes.

Future work may address the
following questions.  Are there classes of games for which these problems
are efficiently solvable?  Can we extend the MIP approach to broader
classes of games?  What results can we obtain for general-sum games?  Note
that just as hardness for symmetric zero-sum games does not imply hardness
for zero-sum games in general (because in the latter the game does not need
to be kept symmetric), in fact hardness for zero-sum games does not imply
hardness for general-sum games (because in the latter the game does not
need to be kept zero-sum).  However, this raises the question of which
solution concept should be used---Nash equilibrium, correlated equilibrium,
Stackelberg mixed strategies, etc.  (All of these coincide in two-player
zero-sum games.)  All in all, we believe that models where a designer has
limited, but not full, control over the game are a particularly natural
domain of study for AI researchers and computer scientists in general, due
to the problems' inherent computational complexity and potential to address
real-world settings.

\section*{Acknowledgements}

We are thankful for support from NSF under awards IIS-1527434,
IIS-0953756, CCF-1101659, and CCF-1337215, ARO under grants
W911NF-12-1-0550 and W911NF-11-1-0332, ERC under StG 639945 (ACCORD), a Guggenheim Fellowship, and a Feodor Lynen research fellowship of the Alexander von Humboldt
Foundation.
This work was done in part while Conitzer was visiting the Simons
Institute for the Theory of Computing.
We thank Martin Bullinger for helpful comments.

\clearpage

\appendix

\section{Omitted Proofs from \secref{sec:zerosum}}

\subsection{Definition and properties of \cyc games}

Consider a triple $(n,d,H)$, where
$n$ is an odd natural number greater or equal to $3$,
$d=(d_1, \ldots, d_n) \in \set{-1,0,1}^n$, and 
$H$ is a rational number greater or equal to $1$. %
The \emph{\cyc game} $C(n,d,H)$ is the matrix game given by $M=(m(i,j))_{i \in I, j \in J}$ with $I=J=[n]$ and  
\[
m(i,j) = 
\begin{cases}
d_i	            &\text{, if } i=j \\
(-1)^{j-i-1} H  &\text{, if } i<j \\
-m(j,i)         &\text{, if } i>j \text.
\end{cases}
\]

\noindent
For example, the \cyc game $C(5,(0,1,0,0,-1),10)$ has payoff matrix
\[
\left(
\begin{array}{rrrrr}
0 & 10 & -10 & 10 & -10 \\	
-10 & 1 & 10 & -10 & 10 \\
10 & -10 & 0 & 10 & -10 \\
-10 & 10 & -10 & 0 & 10 \\
10 & -10 & 10 & -10 & -1
\end{array}
\right) \text.
\]	

An \cyc game $C(n,d,H)$ is symmetric (and equivalent to a weak tournament
game) if and only if $d=(0,\ldots,0)$. Let $\tr(C)$ denote the trace of the payoff matrix of $C$, \ie $\tr(C)=\sum_{i \in [n]} d_i$. We call $C$ \emph{balanced} if $\tr(C)=0$.

\begin{lemma} \label{balanced}
	Let $C = C(n,d,H)$ be an alternating game, let $H \ge (n-1)^2n^3$,  and let $v$ denote the value of $C$.
	\begin{enumerate}\romanenumi
		\item If $C$ is balanced, then $v=0$. \label{lem:i}
		\item If $C$ is not balanced, then $|v| \ge \frac{1}{n^3}$. \label{lem:ii}
		\item Let $\delta > 0$. If a player plays an action with probability $p \notin [ \frac{1}{n}-\frac{(n-1)^2}{H} \delta, \frac{1}{n}+\frac{(n-1)^2}{H} \delta ]$, then the other player can achieve a payoff of at least $\delta -1$.  \label{lem:iii}
\end{enumerate}
\end{lemma}

\begin{proof}
We first prove~\ref{lem:iii}. We prove the case where the row player plays an action with probability $p \notin [ \frac{1}{n}-\frac{(n-1)^2}{H} \delta, \frac{1}{n}+\frac{(n-1)^2}{H} \delta ]$, and the column player obtains a guaranteed payoff. The other case follows by symmetry. Note that $C$ is isomorphic to a game $C'$ in which $C'(i,j)= H$ for $j \in \{ i+1, i+2, \ldots, i+\frac{n-1}{2} \}$, $C'(i,j)= -H$ for $j \in \{ i-1, i-2, \ldots, i-\frac{n-1}{2} \}$ (addition is performed modulo $n$), and $C'(i,i) = C(j,j)$ for some $j$. In particular, if $C$ is balanced then the sum of the entries on the diagonal of $C'$ is zero.
	
First consider the case where $C'$ has only zeros on the diagonal. In this special case we will prove a lower bound of $\delta$ on the colun player's payoff, instead of $\delta-1$. Extending to the general case with the weaker bound will follow easily. So, suppose that the column player can achieve a payoff of no more than $\delta$. We will argue that the probability that the row player places on each row is within the interval $[ \frac{1}{n}-\frac{(n-1)^2}{H} \delta, \frac{1}{n}+\frac{(n-1)^2}{H} \delta ]$. 
	
	Note that for every column $i \in [n]$, we can associate with $i$ a set of $\frac{n-1}{2}$ rows which obtain payoff of $-H$ when the column player plays $i$. We will denote this set of rows by $\underline{i}$, and the set of rows which obtain payoff of $H$ when column $i$ is played by $\overline{i}$. Note that the column player's payoff for playing column $i$ is equal to $H(p(\underline{i})-p(\overline{i})) \le \delta$, for all $i$. Also note that $\overline{i}=\underline{i+\frac{n-1}{2}}$. By applying these (in)equations several times, we have that for any $i,i'$,
	\begin{equation} \label{eqn:every-stretch-close}
		|H(p(\underline{i}) - p(\underline{i'}))| \le (n-1) \delta.
\end{equation}
Now consider the two sets of rows $\underline{j}$ and $\underline{j+1}$. The only difference between the two sets is that row $j+1 \in \underline{j} \setminus \underline{j+1}$ and row $j-\frac{n-1}{2} \in \underline{j+1} \setminus \underline{j}$. Therefore $|H(p(j) - p(j-\frac{n-1}{2}))| \le (n-1)\delta$, or else Equation~\ref{eqn:every-stretch-close} would be violated. Noting that if we start from $j$ and subtract $\frac{n-1}{2}$ iteratively then we will hit every row before returning to row $j$, we can say that $H(p(j)-p(j')) \le (n-1)^2 \delta$ for every row $j'$. Therefore $p(j)$ is within $\frac{(n-1)^2}{H} \delta$ of the average probability placed on each row, which is $\frac{1}{n}$. 

To extend to the case with general diagonal vector $d$, consider the case where $d = (1, \ldots, 1)$, minimizing the payoff for the column player. Then for any outcome, the column player obtains either the same payoff as the $d=(0, \ldots , 0)$ case (if the outcome does not lie on the diagonal), or exactly one less than the $d=(0, \ldots, 0)$ case (if the outcome lies on the diagonal). Therefore the column player can obtain payoff of $\delta -1$ in the general case.	

We now prove~\ref{lem:i}. It will be helpful to associate $C$ with a matrix in the obvious way. We first prove a claim regarding the determinant of $C$, $det(C)$.

\begin{claim} \label{claim:determinant}
 $det(C) = \prod_{i=1}^n(-H+d_i) + \prod_{i=1}^n(H+d_i)$.
\end{claim}

\begin{proof}
Define $\hat{C}$ to be the matrix with $\hat{C}(i,j) = C(i,j-1)+C(i,j)$ (with addition performed modulo $n$). Each row of $\hat{C}$ has only two non-zero entries; in particular, $\hat{C}(i,i)=-H+d_i$ and $\hat{C}(i,i+1)=H+d_i$. Note that $det(\hat{C})=det(C)$. Now to compute $det(\hat{C})$, consider the cofactor expansion using the first column. We obtain
\begin{align}
det(\hat{C}) &= (-H+d_1)det(\hat{C}_{1,1}) + (H+d_n)det( \hat{C}_{n,1}), \label{eqn:cofactor}
\end{align}
where $\hat{C}_{i,j}$ is the matrix obtained by deleting row $i$ and column $j$ from $\hat{C}$. Note that $\hat{C}_{1,1}$ is an upper triangular matrix with $\hat{C}_{1,1}(i,i)=(-H+d_{i+1})$ and $\hat{C}_{n,1}$ is a lower triangular matrix with $\hat{C}_{n,1}(i,i)=(H+d_i)$. Since the determinant of a triangular matrix is the product of its diagonal, Equation~\ref{eqn:cofactor} becomes
\begin{align*}
\begin{split}
det(\hat{C}) &= (-H+d_1)\prod_{i=1}^{n-1}(-H+d_{i+1}) \\&\qquad + (H+d_n)\prod_{i=1}^{n-1}(H+d_i) 
\end{split}\\
&=  \prod_{i=1}^n(-H+d_i) + \prod_{i=1}^n(H+d_i) \qedhere
\end{align*} 
\end{proof}
Next we will use a theorem from \citet{Kapl45a} which states that the value of a completely mixed game is zero if the determinant is zero.\footnote{The theorem is actually more powerful, but this is all that we need.} First we show that $C$ is completely mixed 

\begin{claim}
If $H > (n-1)^2n^3$, then both players'equilibirum strategies place non-zero probability on each available action.
\end{claim}

\begin{proof}
Suppose that the row player places zero probability on some action (the case for the column player is symmetric). Then, by~\ref{lem:iii}, the column player can achieve a payoff of $\frac{H}{n(n-1)^2}-1 > 1$. But by simply placing uniform probability on every row, the column player can do no better than a payoff of $\frac{1}{n}<1$. Therefore the strategy that places zero probability on some row is not a minimax strategy, therefore not an equilibrium strategy.
\end{proof}

We complete the proof by showing that $det(C) = 0$ when $C$ is balanced, which immediately gives us that $v=0$.

\begin{claim}
If $C$ is balanced then $det(C)=0$.
\end{claim}

\begin{proof}
By Claim~\ref{claim:determinant}, it suffices to show that $\prod_{i=1}^n(-H+d_i) + \prod_{i=1}^n(H+d_i)=0$ when $C$ is balanced. Let $B=\prod_{i=1}^n(-H+d_i)$ and $A=\prod_{i=1}^n(H+d_i)=0$. Since $H>1$, $H+d_i>0$ and $-H+d_i<0$ for all $i$. Therefore, since $n$ is odd by definition, $B<0<A$. We will show that $B=-A$. Let $k$ be the number of occurrences of 1  in $d$. Since $C$ is balanced, there must also be $k$ occurrences of $-1$ (and $n-2k$ zeros). To prove the claim, note that
\begin{equation*}
A = (H+1)^k(H-1)^kH^{n-2k},
\end{equation*}
and
\begin{align*}
B &= (-H+1)^k(-H-1)^k(-H)^{n-2k}\\
&= (-1)^k(H-1)^k(-1)^k(H+1)^kH^{n-2k}(-1)^{n-2k}\\
&= (-1)^n(H+1)^k(H-1)^kH^{n-2k}\\
&= -A \qedhere
\end{align*}
\end{proof}

Last we prove~\ref{lem:ii}. Suppose that $C$ is unbalanced and that $d$ contains more 1 entries than $-1$ entries (the opposite case will follow by symmetry). Also suppose that there is exactly one more 1 than $-1$ (since we are showing a lower bound on the value, and the value does not decrease if we increase some entries in the matrix, this will imply the other cases). Let the number of 1 entries be $k$. Suppose for contradiction that $v < \frac{1}{n^3}$.

Suppose that the column player plays some action with probability $p \notin [\frac{1}{n}-\frac{1}{n^2}, \frac{1}{n}+\frac{1}{n^2}]$. Then, by~\ref{lem:iii}, the row player can obtain payoff $\frac{H}{n^2(n-1)^2}-1\ge n-1 \ge 1 > v$, so no such strategy is an equilibrium strategy. Therefore the column player plays every action $i$ with probability $p_i \in  [\frac{1}{n}-\frac{1}{n^2}, \frac{1}{n}+\frac{1}{n^2}]$. There exists some response for the row player such that they receive payoff at least $\frac{1}{n^3}$. In particular, assuming in the worst case that the column player puts maximum probability on the columns with $d_i=-1$ and minimum probability on those with $d_i=1$, playing each row with equal probability gives payoff 
\begin{align*}
k\frac{1}{n}(\frac{1}{n}-\frac{1}{n^2}) - (k-1)\frac{1}{n}(\frac{1}{n}+\frac{1}{n^2}) &= \frac{1}{n^2}-\frac{2k-1}{n^3}\\
&\ge \frac{1}{n^2}-\frac{n-1}{n^3} = \frac{1}{n^3} \qedhere
\end{align*}

\end{proof}

\subsection{Proof of Theorem~\ref{thm:matrix-nec}}

	Membership in coNP is straightforward. To verify a ``no'' instance, it suffices to guess a completion of the game in which the action is not essential.
		
	For hardness, we give a reduction from \textsc{SetCover}. An instance of \textsc{SetCover} is given by a collection $\{ S_1, \ldots, S_n \}$ of subsets of a universe $U$, and an integer $k$; the question
        is whether we can cover $U$ using only $k$ of the subsets. 
We may assume that $k$ is odd (it is always possible to add a singleton
subset with an element not covered by anything else and increase $k$ by $1$). Define an incomplete matrix game $M$ where the row player has $3n-k+1$ actions, and the column player has $2n-k+|U|$ actions. The row player's actions are given by $\{ S_{i,j} \; : \; i \in [n], j \in [2] \} \cup \{ x_i : i \in [n-k]\} \cup \{ r_* \}$.

		\begin{figure} 
		\renewcommand{\arraystretch}{1.8}
		\begin{center}
\footnotesize
\setlength{\tabcolsep}{0.19em}
\scalebox{1.2}{ %
\begin{tabular}{r|c|c|c|c|c||c|c|c|c|c|}
		
		    		  \multicolumn{1}{c}{} & \multicolumn{1}{c}{$c_1$} & \multicolumn{1}{c}{$c_2$} & \multicolumn{1}{c}{$c_3$} & \multicolumn{1}{c}{$c_4$} & \multicolumn{1}{c}{$c_5$} & \multicolumn{1}{c}{$s_1$} & \multicolumn{1}{c}{$s_2$} & \multicolumn{1}{c}{$s_3$} & \multicolumn{1}{c}{$s_4$} & \multicolumn{1}{c}{$s_5$}  \\ \cline{2-11}
		    	    $S_{1,1}$ & $0$ & $\phantom{\scalebox{0.75}[1.0]{\( - \)}}H$ & $\scalebox{0.75}[1.0]{\( - \)}H$ & $\phantom{\scalebox{0.75}[1.0]{\( - \)}}H$ & \phantom{h} $\scalebox{0.75}[1.0]{\( - \)}H$ \phantom{h} & $y$ & $y$ & $0$ & $0$ & $0$  \\ \cline{2-11}
		    	    $S_{1,2}$ & $\{\scalebox{0.75}[1.0]{\( - \)}1, 1 \}$ & $\phantom{\scalebox{0.75}[1.0]{\( - \)}}H$ & $\scalebox{0.75}[1.0]{\( - \)}H$ & $\phantom{\scalebox{0.75}[1.0]{\( - \)}}H$ & $\scalebox{0.75}[1.0]{\( - \)}H$ & $x$ & $x$ & $0$ & $0$ & $0$  \\  \cline{2-11}
		    	    $S_{2,1}$ & $\scalebox{0.75}[1.0]{\( - \)}H$ & $0$ & $\phantom{\scalebox{0.75}[1.0]{\( - \)}}H$ & $\scalebox{0.75}[1.0]{\( - \)}H$ & $\phantom{\scalebox{0.75}[1.0]{\( - \)}}H$ & $0$ & $0$ & $y$ & $y$ & $0$ \\          \cline{2-11}
		    	    $S_{2,2}$ & $\scalebox{0.75}[1.0]{\( - \)}H$ & $\{\scalebox{0.75}[1.0]{\( - \)}1, 1 \}$ & $\phantom{\scalebox{0.75}[1.0]{\( - \)}}H$ & $\scalebox{0.75}[1.0]{\( - \)}H$ & $\phantom{\scalebox{0.75}[1.0]{\( - \)}}H$ & $0$ & $0$ & $x$ & $x$ & $0$\\ \cline{2-11}
		    	    $S_{3,1}$ & $\phantom{\scalebox{0.75}[1.0]{\( - \)}}H$ & $\scalebox{0.75}[1.0]{\( - \)}H$ & $0$ & $\phantom{\scalebox{0.75}[1.0]{\( - \)}}H$ & $\scalebox{0.75}[1.0]{\( - \)}H$ & $0$ & $0$ & $y$ & $0$ & $y$\\          \cline{2-11}
		    	    $S_{3,2}$ & $\phantom{\scalebox{0.75}[1.0]{\( - \)}}H$ & $\scalebox{0.75}[1.0]{\( - \)}H$ & $\{\scalebox{0.75}[1.0]{\( - \)}1, 1 \}$ & $\phantom{\scalebox{0.75}[1.0]{\( - \)}}H$ & $\scalebox{0.75}[1.0]{\( - \)}H$ & $0$ & $0$ & $x$ & $0$ & $x$ \\ \cline{2-11}
		    	    $S_{4,1}$ & ${\scalebox{0.75}[1.0]{\( - \)}}H$ & $\phantom{\scalebox{0.75}[1.0]{\( - \)}}H$ & $\scalebox{0.75}[1.0]{\( - \)}H$ & $0$ & $\phantom{\scalebox{0.75}[1.0]{\( - \)}}H$  & $0$ & $y$ & $0$ & $y$ & 0 \\          \cline{2-11}
		    	    $S_{4,2}$ & ${\scalebox{0.75}[1.0]{\( - \)}}H$ & $\phantom{\scalebox{0.75}[1.0]{\( - \)}}H$ & $\scalebox{0.75}[1.0]{\( - \)}H$ & $\{\scalebox{0.75}[1.0]{\( - \)}1, 1 \}$ & $\phantom{\scalebox{0.75}[1.0]{\( - \)}}H$ &  $0$ & $x$ & $0$ & $x$ & $0$\\ \cline{2-11}
		    		$x_1$     & $\phantom{\scalebox{0.75}[1.0]{\( - \)}}H$ & $\scalebox{0.75}[1.0]{\( - \)}H$ & $\phantom{\scalebox{0.75}[1.0]{\( - \)}}H$ & $\phantom{\scalebox{0.75}[1.0]{\( - \)}}H$ & $\scalebox{0.75}[1.0]{\( - \)}1$ & $0$ & $0$ & $0$ & $0$ & $0$  \\ \hhline{~|=|=|=|=|=#-|-|-|-|-|}
\multicolumn{1}{c|}{$r_*$} & \multicolumn{1}{c|}{$-v$} & \multicolumn{1}{c|}{$-v$} & \multicolumn{1}{c|}{$-v$} & \multicolumn{1}{c|}{$-v$} & \multicolumn{1}{c|}{$-v$} & \multicolumn{1}{c|}{G} & \multicolumn{1}{c|}{G} & \multicolumn{1}{c|}{G} & \multicolumn{1}{c|}{G} & \multicolumn{1}{c|}{G} 

\\       \cline{2-11}
		  \end{tabular}
		  }
		\caption{The incomplete matrix game $M$ used in the proof of \thmref{thm:matrix-nec} for the \textsc{SetCover} instance given by $|U|=5$, $n=4$, $k=3$, $S_1=\{s_1, s_2 \}$, $S_2 = \{s_3, s_4 \}$, $S_3=\{s_3, s_5 \}$, and $S_4 = \{ s_2, s_4 \}$. $L$ lies in the top left, indicated by double lines.
		}
		\label{fig:example}
		\end{center}
		\end{figure}

                Let $L$ denote the restriction of the game to the first
                $N := 2n-k$ columns and $3n-k$ rows. $N$ is odd by assumption that $k$ is odd. We denote the column player's actions in
                this part of the game by $\{ c_1, \ldots, c_{2n-k} \}$. We
                set $m(S_{i,1},c_i)=0$ and $m(S_{i,2},c_i)=\{ -1, 1 \}$ for all $i \in [n]$ and
                $m(x_{i}, c_{n+i})=-1$ for all $i \in [n-k]$. We fill in the remaining entries with
                $H$ and $-H$, where $H = (N-1)^2N^3$, so that $L$
                acts as an \cyc game of size $N$ when exactly one of the actions
                $S_{i,1}$ and $S_{i,2}$ is removed for each~$i$.
	
Let $p_{min}=\frac{1}{N}-\frac{2}{N^3}$ and $p_{max}=\frac{1}{N}+\frac{2}{N^3}$. Here, we denote the columns by $\{ s_1, \ldots, s_{|U|} \}$. For column $s_j$, we set the payoff to be $y = \frac{2n}{p_{min}} >0$ in row $S_{i,1}$ if $s_j \in S_i$, and 0 otherwise. We set the entry in row $S_{i,2}$ to be $x = -\frac{2}{p_{max}}<0$ if $s_j \in S_i$ and 0 otherwise. For row $r_*$, set the entry for each of the first $N$ columns to be $-v \coloneqq \frac{-1}{n^4}$ and all other entries are $G = \frac{2H}{\epsilon^*}$, where $\epsilon^* \coloneqq \frac{1}{N^2(y+x)+1}$. \figref{fig:example} illustrates this construction for a small instance of set cover.
	
	Observe that in any completion $M'=(m'(\cdot,\cdot))$ of $M$, if the column player plays
        only actions from $L$, then for each pair $(S_{i,1},
        S_{i,2})$, the row player will play at most one of them with
        positive probability in equilibrium, since one of the rows will
        weakly dominate the other in $L$ and the column player's strategy will
        have full support in $L$ (see \lemref{balanced}). Also note that on any pair of rows $(S_{i,1}, S_{i,2})$, the total probability mass at an equilibrium in which $r_*$ is not played must lie within the interval $[p_{min},p_{max} ]$. If not, then Lemma~\ref{balanced} says that the column player can obtain payoff at least 1, which violates the value of $M$ (which is at least $-v>-1$ for any completion $M'$). 
	Intuitively, setting the entry $m'(S_{i,2},c_i)$ to $-1$
        corresponds to choosing $S_i$ and setting
        $m'(S_{i,2},c_i)$ to $1$ corresponds to not choosing $S_i$.
		We show that $r_* \in ES(M')$ if, and only if, $M'$ does not correspond to a set cover of size $k$.
	
	First, suppose that we complete the game in a way that corresponds
        to a set cover of size $k$. We show that an equilbrium for $L$ is in fact a quasi-strict equilibrium of the completed game $M'$, therefore excluding $r_*$ from $ES(M')$. If the column player
        plays only columns from $L$, then we have already observed that $M$
        behaves as an \cyc game from the perspective of the row player. By
        \lemref{balanced}, the value of $L$ is zero since the row player has an equal number of (undominated) actions with a $1$ and a $-1$ on the diagonal ($n-k$ of each). $r_*$ is therefore not a best response for the row player. Now let us examine the
        payoff for the column player from playing any action not in~$L$. For any action $s_i$ there is some chosen set $S_j$ that
        covers element $s_i$ and for which there is a probability at least $p_{min}$ that the row player plays $S_{j,1}$, and a probability at most $p_{max}$ that the row player plays $S_{l,2}$ for any $l \not= j$ (of which there are at most $n-1$ possibilities). Thus the column player's payoff for playing $s_i$ is at most $-(n-1)p_{max}x-p_{min}y = -2$, so they will not play $s_i$.

Now suppose that we complete the game in a way that does not correspond to
a set cover of size $k$. There are three cases: 
(1) we set $k$ payoffs to $-1$ but they do not correspond to a set cover; 
(2) we set too many unspecified payoffs to $-1$; or 
(3) we set too few unspecified payoffs to $-1$. 
We show that in each case the essential set contains $r_*$
(or---another possibility in the third case---a set cover in fact exists).
	
\textbf{Case 1:} Exactly $k$ unspecified entries are set to $-1$, but some element $s_i \in U$ is not covered. Suppose for the sake of contradiction that there exists a quasi-strict equilibrium which does not contain $r_*$ in its support. We show that the column player puts at least $\epsilon^*$ probability on columns in $R$. Suppose this was not the case. Then, in particular, column $s_i$ is played with probability less than $\epsilon^*$.  Let $S_j \ni s_i$. We distinguish two cases, and derive a contradiction in each.

\emph{Case 1a:} Column $c_j$ is played with probability at most $\frac{1-\epsilon^*}{N^2}$, which is a distance of more than $\frac{1-\epsilon^*}{N^2}$ from $\frac{1-\epsilon^*}{N}$. Therefore the row player can obtain an expected payoff of at least $\frac{H}{N^2(N-1)^2}-1 = N-1$ whenever the column player plays an action from $L$. Even if the column player gets payoff $x$ every time they play an action from $R$ ($x$ is the best possible payoff for the column player in $R$), the total payoff for the row player is at least
\begin{align*}
(1- \epsilon^*)(N-1)-\epsilon^* x &= (1 - \epsilon^*)(N-1) - \frac{x}{N^2(x+y)+1}\\ & = N-1-\epsilon^*(N+x-1) \\ &> 0 \text,
\end{align*}
which violates the value of the game (which is at most zero, since the column player can guarantee themselves zero by playing only columns from $L$). 

\emph{Case 1b:} Column $c_j$ is played with probability greater than $\frac{1-\epsilon^*}{N^2}$. Then, row $S_{i,2}$ is strictly better than $S_{i,1}$ for the row player, since the difference in payoff between the two rows is greater than
\begin{align*}
\frac{1-\epsilon^*}{N^2}-(x+y)\epsilon^* = 0
\end{align*}
by the choice of $\epsilon^*$. Therefore for every set $S_j$ containing element $s_i$, row $S_{j,1}$ is never played. Therefore the payoff for the column player for playing column $s_i$ is at least $-xp_{min}= 2$, which also violates the value of the game.

Combining these two cases, we conclude that the column player places at least $\epsilon^*$ probability on actions in $R$. But then the row player could obtain a payoff of at least $\epsilon^*G-v = 2H-v>H$ by playing $r_*$, which is a best response to the column player's strategy. This contradicts that $r_*$ is not played at equilibrium.

                 \textbf{Case 2:} More than $k$ entries, say $\ell>k$ entries, are set to $-1$. Note
                 that the value of $M'-r_*$ will be negative; indeed by Lemma~\ref{balanced}, the column player can guarantee herself a payoff of at least $\frac{1}{n^3}$ by playing only columns from $L$. Since $\frac{-1}{n^3} < -v$, the row player will play the pure strategy $r_*$.

                 \textbf{Case 3:} Fewer than $k$ unspecified entries are set to
                 $-1$. There are two possibilities. First, the way of
                 setting the entries corresponds to a set cover. In this
                 case there is necessarily a set cover of size $k$ as well,
                 meaning that $M$ is a ``yes'' instance. Second, the payoffs do
                 not correspond to a set cover. Then some element $s_i$ is uncovered and, by the same argument
                 as for Case 1, the row player puts non-zero probability on action $r_*$.

Together, this shows that $r_*$ is a necessary equilibrium action if and only if there does not exist a set cover of size~$k$.

\subsection{Proof of Theorem~\ref{thm:matrix-pos}}

	Membership in NP is straightforward, as we can guess a completion $M'$ and compute $\es(M')$.

	For hardness, we again give a reduction from \textsc{SetCover}. Given an instance of \textsc{SetCover}, we construct a similar game to the game used in the necessary equilibrium action proof. However, instead of having $n-k$ rows $x_i$, we have $k$ such rows, and choosing a set will now correspond to setting the correspond entry to $1$ instead of $-1$. Note that $L$ now consists of $2n+k$ rows and $N := n+k$ columns. We will assume that $n+k$ is odd (if not, simply add an additional set $S_{n+1}$ containing no elements). We also swap every instance of $x$ and $y$ and we replace $r_*$ with a new row, call it $r'$, with payoff entries 0 for the first $N$ columns and $-G$ for the others.
See \figref{fig:example2} for an example.

		\begin{figure*} 
		\renewcommand{\arraystretch}{1.8}
		\begin{center}
\footnotesize
\setlength{\tabcolsep}{0.19em}
\scalebox{1.2}{
\begin{tabular}{r|c|c|c|c|c|c|c||c|c|c|c|c|}
		
		    		  \multicolumn{1}{c}{} & \multicolumn{1}{c}{$c_1$} & \multicolumn{1}{c}{$c_2$} & \multicolumn{1}{c}{$c_3$} & \multicolumn{1}{c}{$c_4$} & \multicolumn{1}{c}{$c_5$} & \multicolumn{1}{c}{$s_1$} & \multicolumn{1}{c}{$s_2$} & \multicolumn{1}{c}{$s_3$} & \multicolumn{1}{c}{$s_4$} & \multicolumn{1}{c}{$s_5$}  \\ \cline{2-13}
		    	    $S_{1,1}$ & $0$ & $\phantom{\scalebox{0.75}[1.0]{\( - \)}}H$ & $\scalebox{0.75}[1.0]{\( - \)}H$ & $\phantom{\scalebox{0.75}[1.0]{\( - \)}}H$ & \phantom{h} $\scalebox{0.75}[1.0]{\( - \)}H$ \phantom{h} & $\phantom{\scalebox{0.75}[1.0]{\( - \)}}H$ & \phantom{h} $\scalebox{0.75}[1.0]{\( - \)}H$ \phantom{h} & $x$ & $x$ & $0$ & $0$ & $0$  \\ \cline{2-13}
		    	    $S_{1,2}$ & $\{\scalebox{0.75}[1.0]{\( - \)}1, 1 \}$ & $\phantom{\scalebox{0.75}[1.0]{\( - \)}}H$ & $\scalebox{0.75}[1.0]{\( - \)}H$ & $\phantom{\scalebox{0.75}[1.0]{\( - \)}}H$ & $\scalebox{0.75}[1.0]{\( - \)}H$ & $\phantom{\scalebox{0.75}[1.0]{\( - \)}}H$ & $\scalebox{0.75}[1.0]{\( - \)}H$ & $y$ & $y$ & $0$ & $0$ & $0$  \\  \cline{2-13}
		    	    $S_{2,1}$ & $\scalebox{0.75}[1.0]{\( - \)}H$ & $0$ & $\phantom{\scalebox{0.75}[1.0]{\( - \)}}H$ & $\scalebox{0.75}[1.0]{\( - \)}H$ & $\phantom{\scalebox{0.75}[1.0]{\( - \)}}H$ & $\scalebox{0.75}[1.0]{\( - \)}H$ & $\phantom{\scalebox{0.75}[1.0]{\( - \)}}H$ & $0$ & $0$ & $x$ & $x$ & $0$ \\          \cline{2-13}
		    	    $S_{2,2}$ & $\scalebox{0.75}[1.0]{\( - \)}H$ & $\{\scalebox{0.75}[1.0]{\( - \)}1, 1 \}$ & $\phantom{\scalebox{0.75}[1.0]{\( - \)}}H$ & $\scalebox{0.75}[1.0]{\( - \)}H$ & $\phantom{\scalebox{0.75}[1.0]{\( - \)}}H$ & $\scalebox{0.75}[1.0]{\( - \)}H$ & $\phantom{\scalebox{0.75}[1.0]{\( - \)}}H$ & $0$ & $0$ & $y$ & $y$ & $0$\\ \cline{2-13}
		    	    $S_{3,1}$ & $\phantom{\scalebox{0.75}[1.0]{\( - \)}}H$ & $\scalebox{0.75}[1.0]{\( - \)}H$ & $0$ & $\phantom{\scalebox{0.75}[1.0]{\( - \)}}H$ & $\scalebox{0.75}[1.0]{\( - \)}H$ & $\phantom{\scalebox{0.75}[1.0]{\( - \)}}H$ & $\scalebox{0.75}[1.0]{\( - \)}H$ & $0$ & $0$ & $x$ & $0$ & $x$\\          \cline{2-13}
		    	    $S_{3,2}$ & $\phantom{\scalebox{0.75}[1.0]{\( - \)}}H$ & $\scalebox{0.75}[1.0]{\( - \)}H$ & $\{\scalebox{0.75}[1.0]{\( - \)}1, 1 \}$ & $\phantom{\scalebox{0.75}[1.0]{\( - \)}}H$ & $\scalebox{0.75}[1.0]{\( - \)}H$ & $\phantom{\scalebox{0.75}[1.0]{\( - \)}}H$ & $\scalebox{0.75}[1.0]{\( - \)}H$ & $0$ & $0$ & $y$ & $0$ & $y$ \\ \cline{2-13}
		    	    $S_{4,1}$ & ${\scalebox{0.75}[1.0]{\( - \)}}H$ & $\phantom{\scalebox{0.75}[1.0]{\( - \)}}H$ & $\scalebox{0.75}[1.0]{\( - \)}H$ & $0$ & $\phantom{\scalebox{0.75}[1.0]{\( - \)}}H$ & $\scalebox{0.75}[1.0]{\( - \)}H$ & $\phantom{\scalebox{0.75}[1.0]{\( - \)}}H$ & $0$ & $x$ & $0$ & $x$ & 0 \\          \cline{2-13}
		    	    $S_{4,2}$ & ${\scalebox{0.75}[1.0]{\( - \)}}H$ & $\phantom{\scalebox{0.75}[1.0]{\( - \)}}H$ & $\scalebox{0.75}[1.0]{\( - \)}H$ & $\{\scalebox{0.75}[1.0]{\( - \)}1, 1 \}$ & $\phantom{\scalebox{0.75}[1.0]{\( - \)}}H$ & $\scalebox{0.75}[1.0]{\( - \)}H$ & $\phantom{\scalebox{0.75}[1.0]{\( - \)}}H$ &  $0$ & $y$ & $0$ & $y$ & $0$\\ \cline{2-13}
		    		$x_1$     & $\phantom{\scalebox{0.75}[1.0]{\( - \)}}H$ & $\scalebox{0.75}[1.0]{\( - \)}H$ & $\phantom{\scalebox{0.75}[1.0]{\( - \)}}H$ & $\scalebox{0.75}[1.0]{\( - \)}H$ & $-1$ & $\phantom{\scalebox{0.75}[1.0]{\( - \)}}H$ & $\scalebox{0.75}[1.0]{\( - \)}H$ & $0$ & $0$ & $0$ & $0$ & $0$  \\ \cline{2-13}
					$x_2$     & $\scalebox{0.75}[1.0]{\( - \)}H$ & $\phantom{\scalebox{0.75}[1.0]{\( - \)}}H$ & $\scalebox{0.75}[1.0]{\( - \)}H$ & $\phantom{\scalebox{0.75}[1.0]{\( - \)}}H$ & $\scalebox{0.75}[1.0]{\( - \)}H$ & $-1$ & $\phantom{\scalebox{0.75}[1.0]{\( - \)}}H$ & $0$ & $0$ & $0$ & $0$ & $0$  \\ \cline{2-13}
					$x_1$     & $\phantom{\scalebox{0.75}[1.0]{\( - \)}}H$ & $\scalebox{0.75}[1.0]{\( - \)}H$ & $\phantom{\scalebox{0.75}[1.0]{\( - \)}}H$ & $\scalebox{0.75}[1.0]{\( - \)}H$ & $\phantom{\scalebox{0.75}[1.0]{\( - \)}}H$ & $\scalebox{0.75}[1.0]{\( - \)}H$ & $-1$ & $0$ & $0$ & $0$ & $0$ & $0$  \\ \hhline{~|=|=|=|=|=|=|=#-|-|-|-|-|}
\multicolumn{1}{c|}{$r'$} & \multicolumn{1}{c|}{$0$} & \multicolumn{1}{c|}{$0$} & \multicolumn{1}{c|}{$0$} & \multicolumn{1}{c|}{$0$} & \multicolumn{1}{c|}{$0$} & \multicolumn{1}{c|}{$0$} & \multicolumn{1}{c|}{$0$} & \multicolumn{1}{c|}{-G} & \multicolumn{1}{c|}{-G} & \multicolumn{1}{c|}{-G} & \multicolumn{1}{c|}{-G} & \multicolumn{1}{c|}{-G} 

\\       \cline{2-13}
		  \end{tabular}
		  }
		\caption{The incomplete matrix game $M$ used in the proof of \thmref{thm:matrix-pos} for the \textsc{SetCover} instance given by $|U|=5$, $n=4$, $k=3$, $S_1=\{s_1, s_2 \}$, $S_2 = \{s_3, s_4 \}$, $S_3=\{s_3, s_5 \}$, and $S_4 = \{ s_2, s_4 \}$. $L$ lies in the top left, indicated by double lines.
		}
		\label{fig:example2}
		\end{center}
		\end{figure*}
		
 We will show that $r' \in ES(M')$ if and only if $M'$ corresponds to a set cover of size $k$.

	Suppose that the unspecified entries are set corresponding to a set cover. Then by the same argument as in the proof of Theorem~\ref{thm:matrix-nec}, there is an equilibrium lying completely within $L$ which gives payoff exactly zero to both players. Note that every column outside of $L$ gives strictly negative payoff to the column player. However, this is not a quasi-strict equilibrium because the row player has a best response, row $r'$, which is played with zero probability. Observe that there is necessarily some $\epsilon$ probability which the row player can place on $r'$ such that the column player is strictly better off not to change her strategy, and we have a quasi-strict equilibrium. Therefore, $r' \in ES(M')$ for such a completion.

	Now suppose that the unspecified entries are set in a way that does not correspond to a set cover of size~$k$. There are three cases, again mirroring the proof of Theorem~\ref{thm:matrix-nec}. 
	
	\textbf{Case 1:} Exactly $k$ unspecified entries are set to $1$, but some element $s_i \in U$ is uncovered. Consider a quasi-strict equilibrium of $M'-r'$. By a similar argument as in the proof of Theorem~\ref{thm:matrix-nec}, column $s_i$ is played with at least $\epsilon^*$ probability. Therefore, the row player would obtain a payoff of at most
	\begin{equation*}
		-\epsilon^*G = -2H
	\end{equation*}
by playing $r'$ in response to this strategy, which is less than all other entries in the matrix. Therefore the quasi-strict equilibrium of $M'-r'$ is in fact a quasi-strict equilibrium of $M'$, and so $r' \notin ES(M')$.
	
	\textbf{Case 2:} More than $k$ entries are set to $1$. Then $r'$ is strictly dominated by the row player's equilibrium strategy for $L$. In this strategy, the row player gets strictly positive payoff against any column in $L$ (as there are more $1$ entries than $-1$ entries in $L$), and, at worst, $x>-G$ payoff against columns outside of $L$.
	
	\textbf{Case 3:} Fewer than $k$ entries are set to $1$. Then either the completion corresponds to a set cover, in which case there is a set cover of size $k$ also and the instance is a ``yes'' instance, or there is some element uncovered. For the latter case, $r' \notin ES(M')$ by the same argument as for Case 1.
	
	So $r'$ is a possible equilibrium action if, and only if, there exists a set cover of size $k$.

\section{Omitted Proofs from \secref{sec:weak}}

\subsection{Proof of \thmref{thm:weak-pos}}

	Membership in NP is straightforward as we can guess a completion $W'$ of the incomplete weak tournament game and verify that the action is in $\es(W')$.

	For NP-hardness, we provide a reduction from \textsc{Sat}. 
	Let $\varphi = C_1 \land \ldots \land C_m$ be a Boolean formula in conjunctive normal form over a finite set $V= \{v_1, \ldots, v_n\}$ of variables. 

	We define an incomplete weak tournament $W_\varphi=(A,\succ)$ as follows.
	The set $A$ of vertices is given by $A =  \cup_{i=1}^n X_i \cup \set{c_1, \ldots, c_m} \cup \set{d}$, where $X_i = \set{x_i^1,\ldots,x_i^6}$ for all $i \in [n]$. Vertex $c_j$ corresponds to clause $C_j$ and the set $X_i$ corresponds to variable~$v_i$.
	
	Within each set $X_i$, there is a cycle 
$x_i^1 \succ x_i^2 \succ x_i^3 \succ x_i^4 \succ x_i^5 \succ x_i^6 \succ x_i^1$ and an unspecified edge between $x_i^1$ and $x_i^4$. 
	If variable $v_i$ occurs as a positive literal in clause $C_j$, we
        have edges $c_j \succ x_i^3$ and $x_i^5 \succ c_j$. If variable
        $v_i$ occurs as a negative literal in clause $C_j$, we have edges
        $c_j \succ x_i^6$ and $x_i^2 \succ c_j$. Moreover, there is an edge from $c_j$ to $d$ for every $j \in [m]$. For all pairs of vertices for which neither an edge has been defined, nor an unspecified edge declared, we have a tie.  See \figref{fig:construction} for an example.

	We make two observations about the weak tournament~$W_\varphi$. 
	
	\setcounter{observation}{0}
	\begin{observation} \label{obs1}
		For each completion $W$ of $W_\varphi$, we have $d \in \es(W)$ if and only if $\es(W) \cap \set{c_1, \ldots, c_m} = \emptyset$.
	\end{observation}
	
	\begin{observation} \label{obs2}
		For each $i$, there is exactly one unspecified edge within (and thus exactly three possible completions of) the subtournament $W_\varphi|_{X_i}$. If the we set a tie between $x_i^1$ and $x_i^4$, then all Nash equilibria $p$ of the subtournament $W_\varphi|_{X_i}$ satisfy $p(x_i^1)=p(x_i^3)=p(x_i^5)$ and $p(x_i^2)=p(x_i^4)=p(x_i^6)$. If we set $x_i^1 \succ x_i^4$, then every quasi-strict equilibrium $p$ of $W_\varphi|_{X_i}$ satisfies $p(x_i^2)=p(x_i^4)=p(x_i^6)=0$, $p(x_i^5)>p(x_i^1)>p(x_i^3)>0$, and $p(x_i^1)+p(x_i^3)>p(x_i^5)$. 
	By symmetry, setting $x_i^4 \succ x_i^1$ results in quasi-strict equilibria $p$ with $p(x_i^1)=p(x_i^3)=p(x_i^5)=0$, $p(x_i^4)>p(x_i^6)>p(x_i^2)>0$, and $p(x_i^2)+p(x_i^6)>p(x_i^4)$. 
	\end{observation}
	
	We now show hat $\varphi$ is satisfiable if and only if there is a completion $W$ of $W_\varphi$ with $d \in \es(W)$. For the direction from left to right, let $\alpha$ be a satisfying assignment and define the completion $W^\alpha$ of $W_\varphi$ as follows: if $v_i$ is set to ``true'' under $\alpha$, add edge $x_i^1 \succ x_i^4$; otherwise, add edge $x_i^4 \succ x_i^1$.
	The following lemma shows that $d \in \es(W^\alpha)$. 
	
	\begin{lemma} \label{lem:union}
		If $\alpha$ is a satisfying assignment, then $\es(W^\alpha) = \cup_{i \in [n]} \es(W|_{X_i}) \cup \set{d}$.
	\end{lemma}
	
\begin{proof}
	Let $\alpha$ is a satisfying assignment. Define a probability distribution $p$ on $A$ as follows. ($Z$ is a normalizing factor that ensures $\sum_{a \in A} p(a)=1$.)
	\begin{itemize}
		\item $p(d)=\frac{1}{Z}$;
		\item $p(c_j)=0$ for all $j \in [m]$;
		\item if $x_i$ is set to ``true'' under $\alpha$, 
		\[ (p(x_i^1),\ldots,p(x_i^6)) = \left(\frac{3}{Z},0,\frac{2}{Z},0,\frac{4}{Z},0 \right) \text; \]
		\item if $x_i$ is set to ``false'' under $\alpha$, 
		\[ (p(x_i^1),\ldots,p(x_i^6)) = \left(0,\frac{4}{Z},0,\frac{3}{Z},0,\frac{2}{Z} \right) \text. \]
	\end{itemize}
We show that $p$ is a quasi-strict Nash equilibrium of the weak tournament game corresponding to $W^\alpha$. For a vertex $a \in A$, let $u(a,p)$ be the payoff for a player if she plays (pure) strategy $a$ and the other player plays (mixed) strategy $p$. (Since the game is symmetric, the roles of the players are interchangeable and the value of the game is $0$.)
It suffices to show that 
$u(a,p)=0$ for all $a \in \supp(p)$, and $u(a,p)<0$ for all $a \in A \setminus \supp(p)$.

\begin{itemize}
	\item $u(d,p)=0$ since $d$ has no edges to any other vertex in $\supp(p)$. 
	\item For $j \in [m]$, we have $u(c_j,p)=\frac{1}{Z} - t_j^\alpha (\frac{2}{Z}-\frac{4}{Z}) = \frac{1}{Z} (1-2t_j^\alpha)$, where $t_j^\alpha$ is the number of literals in clause $C_j$ that are set to ``true'' under $\alpha$. Since $\alpha$ satisfies $\varphi$, $t_j^\alpha \ge 1$ for all $j \in [m]$. Therefore, $u(c_j,p)<0$. 
	\item If $x_i$ is set to ``true'' under $\alpha$, we have
	\begin{itemize}
		\item $u(x_i^1,p) = u(x_i^3,p) = u(x_i^5,p) = 0$,
		\item $u(x_i^2,p) = \frac{2}{Z} - \frac{3}{Z} <0$,
		\item $u(x_i^4,p) = \frac{4}{Z} - \frac{3}{Z} - \frac{2}{Z} <0$, and
		\item $u(x_i^6,p) = \frac{3}{Z} - \frac{4}{Z} <0$.
	\end{itemize}
	\item If $x_i$ is set to ``false'' under $\alpha$, we have
	\begin{itemize}
		\item $u(x_i^2,p) = u(x_i^4,p) = u(x_i^6,p) = 0$,
		\item $u(x_i^1,p) = \frac{4}{Z} - \frac{3}{Z} - \frac{2}{Z} <0$,
		\item $u(x_i^3,p) = \frac{3}{Z} - \frac{4}{Z} <0$, and
		\item $u(x_i^5,p) = \frac{2}{Z} - \frac{3}{Z} <0$.
	\end{itemize}
\end{itemize}

Since all quasi-strict Nash-equilibria have the same support, we know that $\es(W^\alpha)=\supp(p) = \cup_{i \in [n]} \es(W^\alpha|_{X_i}) \cup \set{d}$. 
\end{proof}

	For the direction from right to left, let $W$ be a completion of $W_\varphi$ with $d \in \es(W)$. Define the assignment $\alpha$ by setting variable $v_i$ to ``true'' if $x_i^1 \succ x_i^4$ and to ``false'' if $x_i^4 \succ x_i^1$. If there is a tie between $x_i^1$ and $x_i^4$, we set the truth value of $v_i$ arbitrarily. We show that $\alpha$ satisfies $\varphi$.
	
	Since $d \in \es(W)$, we know by the first observation above that $c_j \notin \es(W)$ for all $j \in [m]$. Moreover, there are no edges between $d$ and $X_i$, or between $X_i$ and $X_j$ for $i \neq j$. Therefore, $\es(W) = \cup_{i \in [n]} \es(W^\alpha|_{X_i}) \cup \set{d}$.

	Let $p$ be a quasi-strict Nash equilibrium of the weak tournament game corresponding to $W$ and fix $j \in [m]$. As in the proof of \lemref{lem:union}, let $u(a,p)$ be the payoff for a player if she plays strategy $a$ and the other player plays strategy $p$. Since $c_j \notin \es(W)$, we have $u(c_j,p)<0$. The payoff $u(c_j,p)$ can be written as 
	\begin{align*}
		u(c_j,p) &= \sum_{a \in A} p(a) u(c_j,a) \\
			     &= p(d) u(c_j,d) + \sum_{i=1}^n \left(  \sum_{k=1}^6 p(x_i^k) u(c_j,x_i^k) \right)
	\end{align*}
	Since $u(c_j,d)=1$ and $u(c_j,p)<0$, there exists $i \in [n]$ such that $\sum_{k=1}^6 p(x_i^k) u(c_j,x_i^k) <0$. Recalling \obsref{obs2}, we see that the latter only holds if either $C_j$ contains~$x_i$ and $x_i$ is set to ``true'' under $\alpha$ (\ie $x_i^1 \succ x_i^4$ in $W$) or $C_j$ contains~$\lnot x_i$ and $x_i$ is set to ``false'' under $\alpha$ (\ie $x_i^4 \succ x_i^1$ in $W$). Therefore, $\alpha$ satisfies $C_j$.

\subsection{Proof of \thmref{thm:weak-nec}}

Membership in coNP is straightforward as we can guess a completion $W'$ of an incomplete weak tournament game $W$ and verify that the action is in $\es(W')$.

For hardness, we provide a reduction from \textsc{UnSat}. For a CNF formula $\varphi = C_1 \land \ldots \land C_m$ over $V= \{v_1, \ldots, v_n\}$, consider the $W_\varphi = (A,\succ)$ defined in the proof of \thmref{thm:weak-pos}. We construct a new incomplete weak tournament $\hat W_\varphi$ by adding an additional vertex $d'$ to $W_\varphi$. Formally $\hat W_\varphi = (A \cup \set{d'},\succ)$ with $d \succ d'$ and ties between $d'$ and all $a \in A \setminus \set{d}$. See \figref{fig:construction_nec} for an example. 

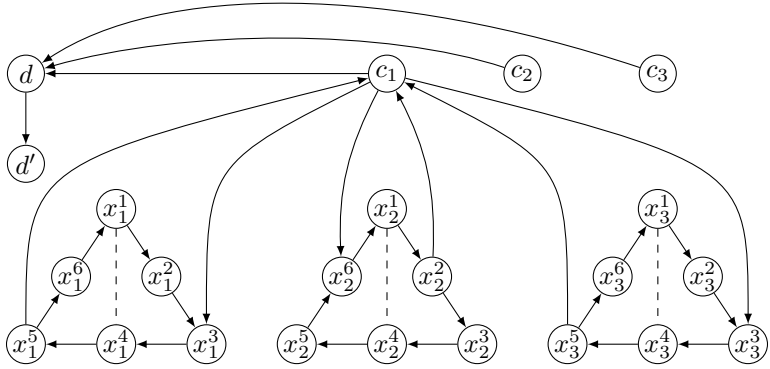
\begin{figure}[tb]
\centering
\begin{tikzpicture}[scale=1.2]
	\tikzstyle{every node}=[circle,draw,minimum size=1.4em,inner sep=0pt]

	\draw(-4,3.5) node(d){$d$};
	
	\draw(-4,2.5) node(d'){$d'$};

	\draw[-latex] (d) -- (d');

	\draw (0,3.5) node(c1){$c_1$} ++(1.5,0) node(c2){$c_2$} ++(1.5,0) node(c3){$c_3$};

		\draw[-latex] (c1) -- (d);
		\draw[-latex] (c2) .. controls  (0,4) and (-2.25,4) .. (d);
		\draw[-latex] (c3) .. controls  (0,4.5) and (-2.5,4.5) .. (d);

	\draw (-3,2) node(x11){$x_1^1$}
	++(0.5,-.75) node(x12){$x_1^2$}
	++(0.5,-.75) node(x13){$x_1^3$}
	++(-1,0)     node(x14){$x_1^4$}
	++(-1,0)     node(x15){$x_1^5$}
	++(0.5,.75)  node(x16){$x_1^6$};
	\foreach \x / \y in {x11/x12,x12/x13,x13/x14,x14/x15,x15/x16,x16/x11}
		{ \draw[-latex,] (\x) -- (\y); }
	\draw[dashed] (x11) -- (x14);

	\draw (0,2) node(x21){$x_2^1$}
	++(0.5,-.75) node(x22){$x_2^2$}
	++(0.5,-.75) node(x23){$x_2^3$}
	++(-1,0)     node(x24){$x_2^4$}
	++(-1,0)     node(x25){$x_2^5$}
	++(0.5,.75)  node(x26){$x_2^6$};
	\foreach \x / \y in {x21/x22,x22/x23,x23/x24,x24/x25,x25/x26,x26/x21}
		{ \draw[-latex,] (\x) -- (\y); }
	\draw[dashed] (x21) -- (x24);

	\draw (3,2) node(x31){$x_3^1$}
	++(0.5,-.75) node(x32){$x_3^2$}
	++(0.5,-.75) node(x33){$x_3^3$}
	++(-1,0)     node(x34){$x_3^4$}
	++(-1,0)     node(x35){$x_3^5$}
	++(0.5,.75)  node(x36){$x_3^6$};
	\foreach \x / \y in {x31/x32,x32/x33,x33/x34,x34/x35,x35/x36,x36/x31}
		{ \draw[-latex,] (\x) -- (\y); }
	\draw[dashed] (x31) -- (x34);

	\draw[-latex] (c1) .. controls (-2,2.5)  .. (x13);
	\draw[-latex] (x15) .. controls (-4,2.5)  .. (c1);

	\draw[-latex] (c1) to [bend right=15] (x26);
	\draw[-latex] (x22) to [bend right=15] (c1);

	\draw[-latex] (c1) .. controls (4,2.5)  .. (x33);
	\draw[-latex] (x35).. controls (2,2.5)  .. (c1);

\end{tikzpicture}
\caption{The weak tournament game $\hat W_\varphi$ for formula $\varphi= C_1 \wedge C_2 \wedge C_3$ with $C_1 = x_1 \lor \lnot x_2 \lor x_3$. Dashed lines indicate unspecified edges. For improved readability, edges connecting $c_2$ and $c_3$ to $X$ have been omitted.}
\label{fig:construction_nec}
\end{figure}

Let $W$ be a completion of $\hat W_\varphi$. Since the edge from $d$ to $d'$ is the only incoming edge of $d'$ in $W$ (and $d'$ does not have any outgoing edges), we have $d' \in \es(\hat W_\varphi)$ if and only if $d \notin \es(\hat W_\varphi)$. 
Furthermore, $\es(W) \cap A = \es(W|_A)$.

Therefore, $d'$ is a necessary $\es$ winner of $\hat W_\varphi$ if and only if $d$ is \emph{not} a possible winner of $W_\varphi$. In the proof of \thmref{thm:weak-pos}, we have shown that the latter is equivalent to $\varphi$ being unsatisfiable.

\subsection{Proof of \propref{prop:continuous}}

Let $W$ be a completion of $W_\varphi$. We have already observed that $d \in \es(W)$ if and only if $\es(W) \cap \set{c_1,\ldots,c_m} = \emptyset$ if and only if $\es(W)=\set{d} \cup \bigcup_{i \in [n]} \es(W|_{X_i})$. 	

We are therefore interested in the equilibria of the weak tournaments $W_\varphi|_{X_i}$ with $i \in [n]$. Fix $i \in [n]$. The incomplete weak tournament $W_\varphi|_{X_i}$ has exactly one unspecified edge, namely $m(x_i^1,x_i^4)=[-1,1]$. For $z \in [-1,1]$, let $W^z$ be the completion of $W_\varphi|_{X_i}$ that has $m(x_i^1,x_i^4)=z$ and, by symmetry, $m(x_i^4,x_i^1)=-z$. 

\begin{lemma}
Let $z \in [-1,1]$ and let $p$ be a quasi-strict Nash equilibrium of $W^z$. 
\begin{itemize}
	\item If $z=0$, then $p(x_i^3)-p(x_i^5)=0$ and $p(x_i^6)-p(x_i^2)=0$.
	
	\item If $z>0$, then $p(x_i^3)-p(x_i^5)<0$ and $p(x_i^6)-p(x_i^2)=0$.
	               
	\item If $z<0$, then $p(x_i^6)-p(x_i^2)<0$ and $p(x_i^3)-p(x_i^5)=0$.
\end{itemize}
\end{lemma}

Therefore, setting $z$ to any value in $(0,1]$ has the same effect as setting $z$ to $1$, and setting $z$ to any value in $[-1,0)$ has the same effect as setting $z$ to $-1$. The proof of \propref{prop:continuous} thus proceeds exactly like the proofs of Theorems \thmref{thm:weak-pos} and \thmref{thm:weak-nec}.

\section{Omitted Proofs from \secref{sec:tournaments}}

\newcommand{\bip}{\mathit{NE}}

In this section, we prove Theorems \ref{thm:posBP} and \ref{thm:necBP}. 
We utilize the one-to-one correspondence between tournament games and directed graphs that are orientations of complete graphs (so-called \emph{tournaments}).

\subsection{Notation}

The following notation will be used throughout this section.
For a tournament $T=(A,\succ)$, a vertex $a \in A$, and a subset $B \subseteq A$ of vertices, we denote by $D_{B,\succ}(a)$ the \emph{dominion} of $a$ in $B$, \ie 
\[D_{B,\succ}(a)=\{\,b\in B\midd a \succ b\,\}\text{,}\] 
and by $\dom_{B,\succ}(a)$ the \emph{dominators} of $a$ in $B$, \ie 
\[\dom_{B,\succ}(a)=\{\,b\in B\midd b \succ a\,\}\text{.}\]
Whenever the dominance relation $\succ$ is known from the context or $B$ is the set of all vertices $A$, the respective subscript will be omitted to improve readability.
We further write $T|_B=(B,\{(a,b)\in B\times B\midd a\succ b\})$ for the restriction of~$T$ to~$B$.

\subsection{Tournament Graphs}
\label{sec:tournaments_app}

An important structural notion in the context of tournaments is that of a \emph{component}. A component is a subset of vertices that bear the same relationship to all vertices not in the set.
\begin{definition}
Let $T=(A,{\succ})$ be a tournament. A nonempty subset $B$ of $A$ is a \emph{component} of $T$ if for all $a\in A\setminus B$, either $B\succ \set{a}$ or $\set{a}\succ B$.  
\end{definition}

For a given tournament $\tilde{T}$, a new tournament~$T$ can be constructed by replacing each vertex with a component. %
\begin{definition}
Let $B_1,\dots,B_k$ be pairwise disjoint sets and consider tournaments $\tilde{T}=(\{1,\dots,k\},\tilde{\succ})$ and $T_1=(B_1,\succ_1)$, \dots, $T_k=(B_k,\succ_k)$. The \emph{product} of $T_1,\dots,T_k$ with respect to $\tilde{T}$, denoted by $\Pi(\tilde{T},T_1,\dots,T_k)$, is the tournament $(A,{\succ})$ such that $A=\bigcup_{i=1}^kB_i$ and for all $b_1\in B_i, b_2\in B_j$, 
\[
b_1\succ b_2 \quad \Leftrightarrow \quad \text{$i=j$ and $b_1\succ_i b_2$, or $i\ne j$ and $i\mathrel{\tilde{\succ}} j$.}
\]
\end{definition}

Our proof also makes use of two special classes of tournaments.   

\begin{definition}
	Let~$n$ be an odd integer and $A=\{a_1, \ldots, a_{n}\}$ an ordered set of size $|A|=n$. The \emph{cyclone on $A$} is the tournament $C_n = (A,\succ)$ such that $a_i\succ a_j$ if and only if $j-i\bmod n\in \set{1,\ldots,\frac{n-1}{2}}$.
\end{definition}

	\figref{fig:c7} depicts a cyclone on 7 vertices.
	
	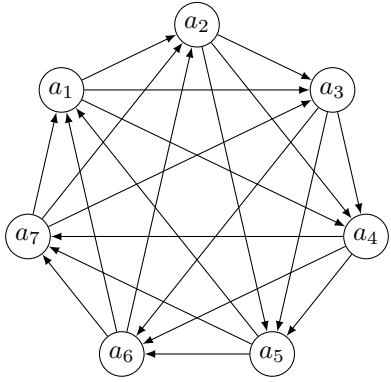
\begin{figure}[tb]
		\centering

		\scalebox{1}{

		\begin{tikzpicture}[scale=1,auto, vertex/.style={circle,draw=black!100,fill=black!00,minimum size=3pt,inner sep=2pt, node distance=0.3},
			component/.style={circle,fill=black!28,draw=black!50,minimum size=3em},	
			componentlight/.style={circle,fill=black!10,draw=black!25,minimum size=3em},
			dottedvertex/.style={vertex,draw=black!99!white,dotted,thick},
			normalvertex/.style={vertex,draw=black!99!white},
			originvertex/.style={vertex,draw=black,double}]

		\path	
			(0,-4) 				node[normalvertex](T4) {$a_6$}
			++(0*51.4285714:2) 	node[normalvertex](T5) {$a_5$}
			++(1*51.4285714:2) 	node[normalvertex](T6) {$a_4$}
			++(2*51.4285714:2) 	node[normalvertex](T7) {$a_3$}
			++(3*51.4285714:2) 	node[normalvertex](T1) {$a_2$}
			++(4*51.4285714:2) 	node[normalvertex](T2) {$a_1$}
			++(5*51.4285714:2) 	node[normalvertex](T3) {$a_7$};

		\foreach \x / \y in {1/2,1/3,1/4,2/3,2/4,2/5,3/4,3/5,3/6,4/5,4/6,4/7,5/6,5/7,5/1,6/7,6/1,6/2,7/1,7/2,7/3}
			{ \draw[latex-] (T\x) -- (T\y); }
	\end{tikzpicture}
	}
	\caption{The tournament $C_7$.}
	\label{fig:c7}
	\end{figure}

\begin{definition}\label{def:babushka}
	Let $k\ge 0$ be an integer. The tournament $B_k$ is inductively defined as follows. $B_0=(A_0,\succ_0)$ is the tournament with a single vertex $x^*$ (and $\succ_0=\emptyset$). Given $B_{k-1}=(A_{k-1},\succ)$, let $B_k = (A_{k-1} \cup \set{a_k,b_k},\succ_k)$, where $a_k, b_k \notin A_{k-1}$ are two new vertices. The edges in $B_k$ are defined such that $\succ_k|_{A_{k-1}} = \succ_{k-1}$ and $A_{k-1} \succ a_k \succ b_k \succ A_{k-1}$.
	  Thus, $A_k = \set{x^*} \cup \bigcup_{i=1}^k \{a_i, b_i\}$.
\end{definition}  
	
	The structure of $B_k$ is illustrated in \figref{fig:babushka}.
	For every $\ell \in \set{0,\ldots,k}$, the set $\set{x^*} \cup \bigcup_{i=1}^\ell \{a_i, b_i\}$ is a component of $B_k$.
	
	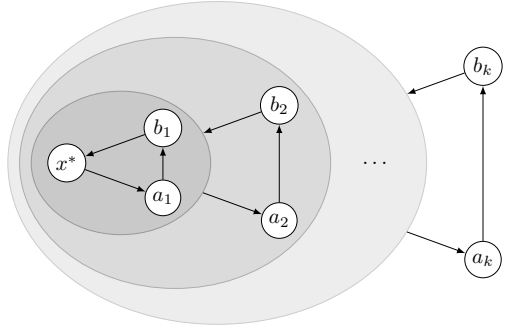
\begin{figure}[tb]
		\centering

		\scalebox{.8}{
			\begin{tikzpicture}[auto,vertex/.style={circle,draw=black!100,fill=black!00,minimum size=3pt,inner sep=2pt, node distance=0.3},
				component/.style={circle,fill=black!28,draw=black!50,minimum size=3em},	
		]
			\draw (0,0) node[ellipse, inner xsep=70, inner ysep=54,fill=black!7,draw=black!20](Ti){}
			+(20: 4.7) node[vertex] (bi){$b_k$}
			+(340:4.7) node[vertex](ai){$a_k$}
			+(0:2.65) node(dots){$\dots$}
			++(180:.7) node[ellipse,inner xsep=52, inner ysep=42,fill=black!14,draw=black!32](T3){}
			++(180:.9) node[ellipse,inner xsep=30, inner ysep=24,fill=black!21,draw=black!40](T2){}
			+(20: 2.8) node[vertex] (b2){$b_2$}
			+(340:2.8) node[vertex](a2){$a_2$}
			++(180:.9) node[vertex](T){$x^*$}			
			+(20: 1.7) node[vertex](b){$b_1$} 
			+(340:1.7) node[vertex](a){$a_1$}
			;

			\foreach \x/\y in {ai/bi,Ti/ai,bi/Ti,a2/b2,T2/a2,b2/T2,a/b,b/T,T/a} {
			\draw[-latex] (\x) -- (\y);
						}
			\end{tikzpicture}
			}
		\caption{Tournament $B_k$. Gray ellipses represent components.}
		\label{fig:babushka}
	\end{figure}

\subsection{Properties of Nash Equilibrium in Tournament Games}

Every tournament game $T$ has unique Nash equilibrium $(p,q)$. This Nash equilibrium is quasi-strict and symmetric, \ie $p=q$. Let $\bip(T)$ denote the equilibrium strategy for one player. If $p=\bip(T)$ and $B \subseteq A$, we let $p(B)$ denote $\sum_{b \in B} p(b)$.  
The following properties will be used repeatedly in our proof.

\begin{lemma}\label{lem:half}
	Let $T=(A,\succ)$ be a tournament.
	\begin{enumerate} \romanenumi
		\item Let $p = \bip(T)$. Then, $p(D(a)) \le p(\dom(a))$ for all $a \in A$ and $p(D(a)) = p(\dom(a))$ if and only if $a \in \bp(T)$. In particular, $p(D(a))\le \frac{1}{2}$ for all $a \in A$. \label{lem:half1}  
		\item Let $a \in A$ and $p = \bip(T|_{A \setminus \set{a}})$. Then, $a \in \bp(T)$ if and only if $p(D(a)) > p(\dom(a))$ if and only if $p(D(a))> \frac{1}{2}$ if and only if $p(\dom(a))< \frac{1}{2}$. \label{lem:half2}
		\item %
		If $a \notin \bp(T)$, then $\bp(T|_{A \setminus \set{a}}) = \bp(T)$. \label{lem:ssp}
	\end{enumerate}
\end{lemma}

\begin{proof}
	See \citet{Lasl97a}, Chapter 6.
\end{proof}

Another important property of Nash equilibrium in tournament games is \emph{composition-consistency}. That is, if a node of a tournament is replaced by a component, the equilibrium probability of a node in the component is the product of the equilibrium probability of the original node and the equilibrium probability that the node gets in the subtournament consisting of the component only. 

\begin{lemma}\label{lem:compcons}
	Let $T = \Pi(\tilde{T},T_1,\dots,T_k)$, where $\tilde{T}=(\{1,\dots,k\},\tilde{\succ})$ and $T_1=(B_1,\succ_1)$, \dots, $T_k=(B_k,\succ_k)$. Let $p_T=\bip(T)$, $p_{\tilde{T}}=\bip(\tilde{T})$, and $p_{T_i}=\bip(T_i)$ for $i=1,\ldots,k$. For a vertex $a \in B_i$, we have
	\[p_T(a) = p_{\tilde{T}}(i) \cdot p_{T_i}(a) \text. \]
\end{lemma}

It can moreover be shown that non-zero equilibrium probabilities are bounded away from zero. 

\begin{lemma}\label{lem:minprob}
	Let $T=(A,\succ)$ be a tournament with $|T|=N$, $a \in A$, and $p=\bip(T)$. If $p(a)>0$, then $p(a)>\frac{1}{N^N}$.
\end{lemma}

\begin{proof}
	\citet{FiRy95b} have shown that the least common denominator of the equilibrium probabilities for a tournament on $N$ vertices is upper bounded $N^{(N-3)/4} \sqrt{N^2-4} < N^N$. It follows that the smallest non-zero equilibrium probability is greater than $\frac{1}{N^N}$. 
\end{proof}

It is easy to compute the equilibrium probabilities for the tournament classes defined at the end of \secref{sec:tournaments_app}.

\begin{lemma}\label{lem:probs}
	\begin{enumerate} \romanenumi
		\item Let $T$ be a cyclone on $A$ with $|A|=n$ and let $p=\bip(T)$. Then, $p(a)=\frac{1}{n}$ for all $a \in A$. 
		\item Let $B_k$ be the tournament defined in \defref{def:babushka} and let $p=\bip(B_k)$. Then, for all $i \le k$, $p(a_i)=p(b_i)=\frac{1}{3^{k-i+1}}$ and $p(x^*)=\frac{1}{3^k}$.
	\end{enumerate}
\end{lemma}

\begin{proof}[Proof sketch]
	Uniform equilibrium probabilities for cyclones are due to symmetry. The equilibrium probabilities for $B_k$ can be easily calculated by applying composition-consistency (\lemref{lem:compcons}) repeatedly.  
\end{proof}

\subsection{Construction of Incomplete Tournament $T_\varphi$}
\label{sec:Tvarphi}

\newcommand{\p}{d}

Our hardness proofs are based on a reduction from \textsc{3Sat}. For a given Boolean formula $\varphi$ in 3-CNF, we are going to construct an incomplete tournament $T_\varphi$. 

Let $\varphi = C_1 \wedge \ldots \wedge C_m$ be a 3-CNF formula over variables $x_1, \ldots, x_n$. Without loss of generality assume that $m$ is odd and that every variable occurs at most once in each clause. 
Let $N = m+3n+1$ and $K = N \log N$.

Define the tournament $T_\varphi = (A,\succ)$ as follows. The set $A$ of vertices is given by 
\[A=C \cup \set{\p,\p'} \cup X \text,\]
where $A = \set{c_1, \ldots, c_m}$ and $X = \bigcup_{i=1}^{2n+1} X_i$. 
For $i=1,\ldots,n$, $X_i=\{x_i^+,x_i^-\}$.
For $i=n+1$, $X_{n+1}$ is given by the set $A_K$ from \defref{def:babushka}, \ie 
\[X_{n+1} = \set{x^*} \cup \bigcup_{i=1}^K \set{a_i,b_i} \text.\]
For $i=n+2,\ldots,2n+1$, $X_i=\{x_i\}$.

For $j=1,\dots,m$, vertex $c_j$ corresponds to clause $C_j$; for $i=1,\ldots,n$, the pair of vertices $X_i=\{x_i^+,x_i^-\}$ corresponds to variable $x_i$.

We now define the edges of $T_\varphi$. The subtournament $T_\varphi|_C$ is a cyclone on $C$. Furthermore, $\set{\p,\p'} \succ C$ and $\p \succ \p'$.
The subtournament $T_\varphi|_X$ has components $X_1, \ldots, X_{2n+1}$, which are connected by a cyclone of size $2n+1$. 
For $i=1,\ldots,n$, the subtournament $T_\varphi|_{X_i}$ is incomplete, \ie there is no edge between $x_i^+$ and $x_i^-$.
For $i=n+1$, the subtournament  $T_\varphi|_{X_{n+1}}$ is given by $B_K$.
(For $i=n+2,\ldots,2n+1$, $T_\varphi|_{X_i}$ consists of a single vertex $x_i$ and thus has no edges.)

We go on to define the edges between $\set{\p,\p'}$ and $X$.
The edges incident to $\p$ are defined such that
\begin{align*}
 &D_X(\p) = \bigcup_{i=n+2}^{2n+1} X_i \cup \set{a_1,\ldots,a_K}   \text, \quad \text{and} \\
 &\dom_X(\p) = \bigcup_{i=1}^{n} X_i \cup \set{x^*,b_1,\ldots,b_K} \text.
\end{align*}
The edges incident to $\p'$ are defined such that
\begin{align*}
 &D_X(\p') = \bigcup_{i=1}^{n} X_i \cup \set{b_1,\ldots,b_K}  \text, \quad \text{and} \\
 &\dom_X(\p') = \bigcup_{i=n+2}^{2n+1} X_i \cup \set{x^*,a_1,\ldots,a_K}  \text.
\end{align*}
Thus, we have $D_X(\p') \cup \set {x^*} = \dom_X(d)$.

It remains to define the edges between $C$ and $X$. 
\begin{itemize}
	\item For $i=1,\ldots,n$, consider clause $C_j$ ($j \in {1,\ldots,m}$).
	\begin{itemize}
		\item If $x_i$ occurs as a positive literal in $C_j$, then $x_i^+ \succ c_j \succ x_i^-$;
		\item if $x_i$ occurs as a negative literal in $C_j$, then $x_i^- \succ c_j \succ x_i^+$; and
		\item if $x_i$ does not occur in $C_j$, then $\set{c_j} \succ \set{x_i^+,x_i^-}$.
	\end{itemize}
	\item For $i=n+1$, $C \succ X_{n+1}$.
	\item For $i=n+2,\ldots,2n+1$, $X_i \succ C$.
\end{itemize}

This completes the definition of $T_\varphi$.

\subsection{Properties of $T_\varphi$}

The only edges in $T_\varphi$ that are \emph{not} specified are the edges between $x_i^+$ and $x_i^-$, for $i=1,\ldots,n$. Every completion of $T_\varphi$ corresponds to an assignment. 

\begin{definition}
	For assignment $\alpha$, let $T_\varphi^\alpha \in [T_\varphi]$ be the completion of $T_\varphi$ with 
	\begin{itemize}
		\item $x_i^+ \succ x_i^-$ if $\alpha(x_i)=\text{true}$, and
		\item $x_i^- \succ x_i^+$ if $\alpha(x_i)=\text{false}$.
	\end{itemize}
\end{definition} 

\begin{lemma}\label{lem:sat}
	An assignment $\alpha$ satisfies $\varphi$ if and only if $\bp(T_\varphi^\alpha) \subseteq X$.
\end{lemma}

\begin{proof}
	For the direction from right to left, assume that $\bp(T_\varphi^\alpha) \subseteq X$ and let $p=\bip(T_\varphi^\alpha)$. Since $T_\varphi^\alpha|_X$ is a composed tournament with summary $C_{2n+1}$, composition-consistency implies $p(X_i)=\frac{1}{2n+1}$ for all $i=1,\ldots,n$. In particular, for $i=1,\ldots,n$, 
	\[ p(x_i^+) = 
	\begin{cases}
		\frac{1}{2n+1}, &\text{if } \alpha(x_i)=\text{true} \\
		0,               &\text{if } \alpha(x_i)=\text{false}
	\end{cases}
	\]
	and
	\[ p(x_i^-) = 
	\begin{cases}
		\frac{1}{2n+1}, &\text{if } \alpha(x_i)=\text{false} \\
		0,               &\text{if } \alpha(x_i)=\text{true}
	\end{cases}
	\]
Now consider $D_X(c_j)$, where $j \in \set{1,\ldots,m}$. Since $\bp(T_\varphi^\alpha) \subseteq X$ and $c_j \notin \bp(T_\varphi^\alpha)$, \lemref{lem:half} implies $p(D_X(c_j))<\frac{1}{2}$. Let $f_j^\alpha$ denote the number of literals in clause $C_j$ that are set to ``false'' under assignment $\alpha$. We have 
\[ p(D_X(c_j)) = \frac{1}{2n+1} (n-2+f_j^\alpha)\text. \]
Therefore, $p(D_X(c_j))<\frac{1}{2}$ is equivalent to $f_j^\alpha < \frac{5}{2}$. Consequently, at least one literal in $C_j$ is set to ``true'' under $\alpha$. Since this holds for all $j \in \set{1,\ldots,m}$, assignment $\alpha$ satisfies $\varphi$.

For the direction from left to right, assume that $\alpha$ satisfies~$\varphi$ and let $p = \bip(T_\varphi^\alpha|_X)$. We show that $p$ (appended with $p(a)=0$ for all $a \in A \setminus X$) is also the Nash equilibrium of~$T_\varphi^\alpha$. According to \lemref{lem:half}, it is sufficient to show that 
\begin{align}\label{smaller}
	p(D_X(a)) < \frac{1}{2} \quad \text{for all $a \in A \setminus X$.}
\end{align} 
Since $A \setminus X = C \cup \set{\p,\p'}$, there are three cases to consider. 

If $a \in C$, \ie $a = c_j$ for some $j \in \set{1,\ldots,m}$, we can apply the same reasoning as above to get $p(D_X(c_j)) < \frac{1}{2}$ if and only if $f_j^\alpha < \frac{5}{2}$, and the latter holds because $\alpha$ satisfies~$C_j$.

For $\p$, we have 
\begin{align*}
	p(D_X(\p)) &= p\left(\bigcup_{i=n+2}^{2n+1} X_i \cup D_{X_{n+1}}(\p)\right) \\
				 &= \frac{1}{2n+1} \left(n+ \frac{1-\frac{1}{3^K}}{2}\right) \\
				 &< \frac{1}{2n+1} \left(n+ \frac{1}{2}\right) = \frac{1}{2}\text.
\end{align*}
Similarly, we get 
\begin{align*}
	p(D_X(\p')) &= p\left(\bigcup_{i=1}^{n} X_i \cup \dom_{X_{n+1}}(\p)\right) < \frac{1}{2}\text.
\end{align*}
This proves (\ref{smaller}). Consequently, we have $p = \bip(T_\varphi^\alpha)$ and, in particular, $\bp(T_\varphi^\alpha) \subseteq X$.
\end{proof}

\begin{lemma}\label{lem:unsat}
	An assignment $\alpha$ does not satisfy $\varphi$ if and only if $\p \in \bp(T_\varphi^\alpha)$.
\end{lemma}

\begin{proof}
	The direction from right to left follows immediately from \lemref{lem:sat}. For the direction from left to right, assume that $\alpha$ does not satisfy $\varphi$ and let $p=\bip(T_\varphi^\alpha)$. \lemref{lem:sat} implies that $\bp(T_\varphi^\alpha) \not\subseteq X$, \ie $p(X)<1$. We prove that $\p \in \bp(T_\varphi^\alpha)$ in two steps.

	\paragraph{Step 1.} We first show that 
	\begin{align} \label{eq:step1}
		\bp(T_\varphi^\alpha) \cap \set{\p,\p'} \neq \emptyset \text.
	\end{align}
	
	Assume for contradiction that $\bp(T_\varphi^\alpha) \cap \set{\p,\p'} = \emptyset$ and let $p=\bip(T_\varphi^\alpha)$. Then,
	\begin{align} \label{eq:cx}
	\begin{split}
		p(D(\p)) + p(D(\p')) &= p(X \setminus \set{x^*}) + 2p(C)\\
							 &= p(X) + p(C) + p(C) - p(x^*) \\
							 &= 1 + p(C) - p(x^*) \text,
	\end{split}
	\end{align}
	where the last equality is due to $p(X \cup C \cup \set{\p,\p'})=1$ and the assumption that $p(\set{\p,\p'})=0$. 
	
	We are going to show that $p(C) > p(x^*)$. In order to do so, we first derive a lower bound on $p(C)$. \lemref{lem:sat} implies that $p(C)>0$. Let $c \in C$ be a vertex with $p(c)>0$. Since $\p,\p' \notin \bp(T_\varphi^\alpha)$, the equilibrium of $T_\varphi^\alpha$ coincides with the equilibrium of $T_\varphi^\alpha|_{X \cup C}$. Moreover, the set $X_{n+1}$ is a component of $T_\varphi^\alpha|_{X \cup C}$. Therefore, the equilibrium probability of $c$ in $T_\varphi^\alpha$ is equal to the equilibrium probability of $c$ in the tournament in which $\p$ and $\p'$ have been deleted and component $X_{n+1}$ is replaced by a single node. Since the latter tournament has $N = m+3n+1$ nodes, \lemref{lem:minprob} implies that 
	\begin{align}\label{eq:cbound}
		p(C) \ge p(c) > \frac{1}{N^N} \text.
	\end{align} 
	
	We now derive an upper bound on $p(x^*)$. We again use the fact that $X_{n+1}$ is a component of $T_\varphi^\alpha|_{X \cup C}$. Composition-consistency (\lemref{lem:compcons}) implies that $p(x^*) = p(X_{n+1}) \cdot p'(x^*)$, where $p'(x^*)$ is the equilibrium probability of $x^*$ in the subtournament $T_\varphi^\alpha|_{X_{n+1}}$. By \lemref{lem:probs}, $p'(x^*)=\frac{1}{3^K}$. Recalling that $K= N \log N$ was chosen such that $3^K > N^N$, we get 
	\begin{align}\label{eq:xbound}
		p(x^*) = p(X_{n+1}) \cdot p'(x^*) \le p'(x^*) = \frac{1}{3^K} < \frac{1}{N^N} \text.
	\end{align}
	
	Plugging (\ref{eq:cbound}) and (\ref{eq:xbound}) into (\ref{eq:cx}) results in the inequality $p(D(\p)) + p(D(\p')) >1$, which in turn implies that $p(D(\p)) > \frac{1}{2}$ or $p(D(\p')) > \frac{1}{2}$. By \lemref{lem:half}, either case leads to a contradiction. We have thus proven (\ref{eq:step1}).

	\paragraph{Step 2.}
	We now show that $\p \in \bp(T_\varphi^\alpha)$. Assume for contradiction that $\p \notin \bp(T_\varphi^\alpha)$. Then (\ref{eq:step1}) implies $\p' \in \bp(T_\varphi^\alpha)$. Moreover, \lemref{lem:half}~\ref{lem:ssp} implies that $\bp(T_\varphi^\alpha|_{A \setminus \set{\p}}) = \bp(T_\varphi^\alpha)$; in particular, $\p' \in \bp(T_\varphi^\alpha|_{A \setminus \set{\p}})$.
	
	Consider $T' = T_\varphi^\alpha|_{A \setminus \set{\p,\p'}}$ and let $p = \bip(T')$. Since $\p' \in \bp(T_\varphi^\alpha|_{A \setminus \set{\p}})$, \lemref{lem:half}~\ref{lem:half2} implies
	\begin{align}\label{eq:>1/2}
		p(D_{A \setminus \set{\p,\p'}}(\p'))>\frac{1}{2} \text.
	\end{align}
	Observe that $D_{A \setminus \set{\p,\p'}}(\p') = C \cup \bigcup_{i=1}^n X_i \cup \dom_{X_{n+1}}(x^*) = \dom_{A \setminus \set{\p,\p'}}(x^*)$. Therefore, (\ref{eq:>1/2}) implies 
	\begin{align} \label{eq:probx}
			p(\dom_{A \setminus \set{\p,\p'}}(x^*))>\frac{1}{2} \text,
	\end{align}
	which by \lemref{lem:half}~\ref{lem:half2} is equivalent to $x^* \notin \bp(T')$.
	Recall that $X_{n+1}$ is a component of $T'$ and that $\bp(T'|_{X_{n+1}}) = X_{n+1}$. Composition-consistency (\lemref{lem:compcons}) therefore implies that either all or none of the vertices in $X_{n+1}$ are contained in $\bp(T')$. Since we have already shown that $x^* \notin \bp(T')$, we actually have $x \notin \bp(T')$ for all $x \in X_{n+1}$. In other words, $p(X_{n+1})=0$. 
	
	Now consider vertex $x_{2n+1}$. Summing up the equilibrium probabilities of the dominion of this vertex, we get
	\begin{align*}
		p(D_{A \setminus \set{\p,\p'}}(x_{2n+1})) &= p(C \cup \bigcup_{i=1}^n X_i) \\
												  &= p(C \cup \bigcup_{i=1}^n X_i) + 0 \\
												  &= p(C \cup \bigcup_{i=1}^n X_i) + p(\dom_{X_{n+1}}(x^*))\\
												  &= p(\dom_{A \setminus \set{\p,\p'}}(x^*)) >\frac{1}{2} \text,
	\end{align*}
	where the inequality is due to (\ref{eq:probx}).
	By \lemref{lem:half}~\ref{lem:half1}, this is a contradiction. We have thus shown that $\p \in \bp(T_\varphi^\alpha)$, completing the proof.  
\end{proof}

\subsection{Complexity Results}

We are now ready to prove that computing necessary $\bp$ winners is intractable. 

\newtheorem*{thm:necBP}{Theorem \ref{thm:necBP}}
\begin{thm:necBP}
The necessary $\bp$ winner problem (in tournament games) is coNP-complete.
\end{thm:necBP}

\begin{proof}
	Membership in coNP is straightforward: Given an incompletely specified tournament and a distinguished vertex, guess a completion and verify that the vertex is not contained in the bipartisan set of the completion.
	
	Hardness follows from a reduction from \textsc{3Sat}. Given a 3CNF $\varphi$, let $T_\varphi$ be the incomplete tournament defined in \secref{sec:Tvarphi}. Recall that there is a one-to-one correspondence between assignments and completions of $T_\varphi$. \lemref{lem:unsat} states that an assignment $\alpha$ does not satisfy $\varphi$ if and only if $\p \in \bp(T_\varphi^\alpha)$.
	Therefore, $\varphi$ is unsatisfiable if and only if $\p \in \bp(T_\varphi^\alpha)$ for all assignments $\alpha$. The latter is equivalent to $\p$ being a necessary $\bp$ winner of $T_\varphi$.
\end{proof}

\newcommand{\T}{\hat T_\varphi}

For proving hardness of the possible $\bp$ winner problem, we use a slightly modified version of the incomplete tournament $T_\varphi$. For a 3CNF $\varphi$, let $\T=(A,\succ')$ be the incomplete tournament that is identical to $T_\varphi=(A,\succ)$, except that all edges incident to $\p$ are reversed. That is, $\p \succ' a$ if and only if $a \succ \p$ for all $a \in A \setminus \set{\p}$, and $a \succ' b$ if and only if $a \succ b$ for all $a,b \in A \setminus \set{\p}$. The set of unspecified edges in $\T$ is identical to the set of unspecified edges in $T_\varphi$, and completions correspond to assignments in the usual way. The following lemma is the core of our hardness proof. 

\begin{lemma}\label{lem:unsat'}
	An assignment $\alpha$ satisfies $\varphi$ if and only if $\p \in \bp(\T^\alpha)$.
\end{lemma}

\begin{proof}
	\lemref{lem:unsat} shows that $\alpha$ satisfies $\varphi$ if and only if $\p \notin \bp(T_\varphi^\alpha)$. It thus suffices to show that $\p \notin \bp(T_\varphi^\alpha)$ if and only if $\p \in \bp(\T^\alpha)$.  
	Observe that the tournaments $T_\varphi^\alpha|_{A \setminus \set{\p}}$ and $\T^\alpha|_{A \setminus \set{\p}}$ are identical and let 
	\[p = \bip(T_\varphi^\alpha|_{A \setminus \set{\p}}) = \bip(\T^\alpha|_{A \setminus \set{\p}})\text.\] 
	With \lemref{lem:half}~\ref{lem:half2}, we get
	\begin{align*}
		\p \notin \bp(T_\varphi^\alpha) \quad &\Leftrightarrow \quad  p(\dom_{T_\varphi^\alpha}(\p)) > \frac{1}{2} \\
									          &\Leftrightarrow \quad  p(D_{\T^\alpha}(\p)) > \frac{1}{2} \\
									          &\Leftrightarrow \quad  \p \in \bp(\T^\alpha) \text.
	\end{align*}
\end{proof}

\newtheorem*{thm:posBP}{Theorem \ref{thm:posBP}}
\begin{thm:posBP}
The possible $\bp$ winner problem (in tournament games) is NP-complete.
\end{thm:posBP}

\begin{proof}
	Membership in NP is straightforward: Given an incompletely specified tournament and a distinguished vertex, guess a completion and verify that the vertex is contained in the bipartisan set of the completion.
	
	Hardness follows from a reduction from \textsc{3Sat}. Given a 3CNF $\varphi$, let $\T$ be the incomplete tournament defined above. \lemref{lem:unsat'} states that an assignment $\alpha$ satisfies $\varphi$ if and only if $\p \in \bp(\T^\alpha)$.
	Therefore, $\varphi$ is satisfiable if and only if $\p \in \bp(\T^\alpha)$ for some assignment $\alpha$. The latter is equivalent to $\p$ being a possible $\bp$ winner of $\T$.
\end{proof}

\end{document}